\documentclass[sigconf, nonacm]{acmart}

\usepackage{amsmath,amsfonts} 
\usepackage{algorithm,algorithmic}
\usepackage{graphicx}
\usepackage{textcomp}
\usepackage{dsfont}
\usepackage{xcolor} 

\usepackage{graphicx}
\usepackage{makecell}

\usepackage{subfigure}
\usepackage{tabu}

\newtheorem{theorem}{\bf Theorem}[section]
\newtheorem{Definition}{\bf Definition}
\newtheorem{lemma}{Lemma}[section]

\newcommand\vldbdoi{XX.XX/XXX.XX}
\newcommand\vldbpages{XXX-XXX}
\newcommand\vldbvolume{14}
\newcommand\vldbissue{1}
\newcommand\vldbyear{2020}
\newcommand\vldbauthors{\authors}
\newcommand\vldbtitle{\shorttitle} 
\newcommand\vldbavailabilityurl{URL_TO_YOUR_ARTIFACTS}
\newcommand\vldbpagestyle{plain} 

\begin{document}

\title{LDP-FPMiner: FP-Tree Based Frequent Itemset Mining with Local Differential Privacy}

\author{Zhili Chen}
\affiliation{%
  \institution{Shanghai Key Laboratory of Trustworthy Computing,\\
East China Normal University}
  \city{Shanghai}
  \country{China}
}
\email{zhlchen@sei.ecnu.edu.cn}

\author{Jiali Wang}
\affiliation{%
  \institution{School of Computer Science and Technology,\\ Anhui University}
  \city{Hefei}
  \country{China}
}
\email{992130182@qq.com}

\begin{abstract}
Data aggregation in the setting of local differential privacy (LDP) guarantees strong privacy by providing plausible deniability of sensitive data. Existing works on this issue mostly focused on discovering heavy hitters, leaving the task of frequent itemset mining (FIM) as an open problem. To the best of our knowledge, the-state-of-the-art LDP solution to FIM is the SVSM protocol proposed recently. The SVSM protocol is mainly based on the padding and sampling based frequency oracle (PSFO) protocol, and regarded an itemset as an independent item without considering the frequency consistency among itemsets.

In this paper, we propose a novel LDP approach to FIM called LDP-FPMiner based on frequent pattern tree (FP-tree). Our proposal exploits frequency consistency among itemsets by constructing and optimizing a noisy FP-tree with LDP. Specifically, it works as follows. First, the most frequent items are identified, and the item domain is cut down accordingly. Second, the maximum level of the FP-tree is estimated. Third, a noisy FP-tree is constructed and optimized by using itemset frequency consistency, and then mined to obtain the $k$ most frequent itemsets. Experimental results show that the LDP-FPMiner significantly improves over the state-of-the-art approach, SVSM, especially in the case of a high privacy level.
\end{abstract}

\maketitle

\pagestyle{\vldbpagestyle}
\begingroup\small\noindent\raggedright\textbf{PVLDB Reference Format:}\\
\vldbauthors. \vldbtitle. PVLDB, \vldbvolume(\vldbissue): \vldbpages, \vldbyear.\\
\href{https://doi.org/\vldbdoi}{doi:\vldbdoi}
\endgroup
\begingroup
\renewcommand\thefootnote{}\footnote{\noindent
This work is licensed under the Creative Commons BY-NC-ND 4.0 International License. Visit \url{https://creativecommons.org/licenses/by-nc-nd/4.0/} to view a copy of this license. For any use beyond those covered by this license, obtain permission by emailing \href{mailto:info@vldb.org}{info@vldb.org}. Copyright is held by the owner/author(s). Publication rights licensed to the VLDB Endowment. \\
\raggedright Proceedings of the VLDB Endowment, Vol. \vldbvolume, No. \vldbissue\ %
ISSN 2150-8097. \\
\href{https://doi.org/\vldbdoi}{doi:\vldbdoi} \\
}\addtocounter{footnote}{-1}\endgroup

\ifdefempty{\vldbavailabilityurl}{}{
\vspace{.3cm}
\begingroup\small\noindent\raggedright\textbf{PVLDB Artifact Availability:}\\
The source code, data, and/or other artifacts have been made available at \url{\vldbavailabilityurl}.
\endgroup
}

\section{Introduction}
\label{intro}
Differential privacy (DP) has become the $de\ facto$ standard for privacy protection. It was named one of the ten breakthrough technologies in 2020 by the MIT technology review\cite{MITreview}. Generally, there are two types of differential privacy - centralized differential privacy (CDP)\cite{a7} and local differential privacy (LDP)\cite{ldp}, and the focus of this work is the local setting. The most typical LDP protocols \cite{a1,a2,a8,privtrie,privkv,privacyatscal,b1,b2,b3,b4,a12} enable users to randomly perturb their inputs. This guarantees strong privacy without relying on a trusted third party by providing plausible deniability of sensitive data. In practice, LDP has many compelling applications deployed by Apple\cite{apple,applenewwords}, Google\cite{rappor,rapporunknow}, Microsoft\cite{microsoft} and Alibaba\cite{alibaba}, and so on.

As the development of data analysis, privacy issues have drawn more and more attention. Over the past 30 years, data in various fields have increased on a large scale. Such massive amounts of data might have potential correlations (or patterns), which can be extracted or mined for more interesting knowledge \cite{datamining}. As a core data mining task, frequent itemset (or pattern) mining (FIM) plays an essential role in mining association rules\cite{apriori,apr}. However, it also poses a threat to user privacy\cite{a10}. An attacker with strong background information may learn private information from the unprotected itemsets discovered. For this reason, extensive studies have been conducted for the task of privacy-preserving frequent itemset mining (PPFIM)\cite{ppfim,b5,b6}. Especially, differentially private schemes for frequent itemset mining have come to the fore\cite{a3,a4,a5,a6,a2}.

In this paper, we study the task of discovering top-$k$ itemsets over sensitive transactional (or set-valued) data in the context of LDP. Specifically, there are $n$ users, whose transaction $t$ is a subset of $d$ distinct items, denoted by the item domain $\mathcal{X} = \{x_1,x_2,...,x_d\}$. An untrusted analyst (or aggregator) wants to discover the $k$ most frequent itemsets with a given privacy budget $\epsilon$, which measures the scale of privacy provided. This is more challenging even when one just tries to find heavy hitters, and one alternative to address this problem is to encode each transaction as an input, i.e., a value in the power set $\mathfrak{p}(\mathcal{X})$, and then apply existing frequency oracle protocols, such as RAPPOR\cite{rappor} and OLH\cite{a8}, to privately collect estimations. However, in this particular case, the exponential size of the domain $\mathfrak{p}(\mathcal{X})$ may inject considerable noise that results in a poor accuracy.

Meanwhile, the heterogeneous number of items that users hold in the transactional data setting makes the task more complicated. To deal with this, the \emph{padding and sampling (PS)} technique is widely used in the literature, i.e., the user pads her transaction $t$ with dummy items to a specified size $l$ and randomly selects one item as her input, denoted by $PS_{l}(t)$. However, the optimal selection of $l$ is non-trivial task. For instance, in \cite{a1,privset,gu2019pckv}, they suppose each user has at most predefined $L$ items. Ye et al.\cite{privkv} convert the key-value set of each user to its length $d$ binary form, which does not work well for a large domain $d$\cite{gu2019pckv}. The baseline strategy of setting suitable $l$\cite{a2} is to use the 90th percentile of the length of inputs collected privately in the context of LDP.

To the best of our knowledge, the state-of-the-art solution is the Set-Value itemSet Mining (SVSM) protocol \cite{a2}, which discovers top-$k$ itemsets based on the \emph{padding and sampling based frequency oracle (PSFO)} protocol. The core idea of SVSM is to construct a pony-size domain set $|\emph{IS}| = 2k$ of potential itemsets likely to be frequent according to their guessing frequencies, then encode each itemset as one singleton and apply the PSFO protocol with the domain $\emph{IS}$ to privately collect estimations. 
However, the SVSM does not consider the consistency among itemsets, leaving considerable room for the performance improvement.

Inspiringly, we introduce the structure of frequent-pattern tree (FP-tree)\cite{fp} to the solution of FIM problem with LDP for the first time. FP-tree can be used to effectively discover frequent patterns in the traditional non-privacy setting with mild computational cost. In our context, we use an FP-tree constructed with LDP to exploit the frequency consistency among itemsets to improve the data utility. The post-processing property guarantees that itemset mining over the noisy FP-tree do not disclose the privacy as well. Specifically, the noisy FP-tree is constructed in breadth-first (BF) order, and instead of dealing with the heterogeneous number of items, each transaction can be converted into one prefix of the FP-tree. And, to allocate the privacy budget, we approximate the maximum level $M$ of the tree by setting it as the 80th percentile of length distribution of users. Although the large size of the domain at each level increases the noise as well as the cost, we propose a pruning algorithm to effectively cut down the domain into a small fixed size for accuracy improvement.

Summarily, our main contributions are as follows.
\begin{itemize}
\item We propose a novel approach called LDP-FPMiner to discover $k$ most frequent itemsets in LDP setting based on FP-tree for the first time.
\item We design an algorithm that constructs an FP-tree with LDP in breadth-first order and optimize the noisy FP-tree by exploring frequency consistency.
\item Experimental results on both synthetic and real-world datasets show a significant performance improvement over the state-of-the-art SVSM.
\end{itemize}

\noindent\textbf{Roadmap.} The remainder of this paper is organized as follows. Section \ref{preliminaries} gives the preliminaries. Section \ref{problem} introduces the problem setting and the state-of-the-art approach. We present our approach and conduct theoretical analysis in Section \ref{fpmine sum}. In Section \ref{optimizations}, we optimize the proposed scheme in several ways. The experimental results are presented in Section \ref{experiment}. Section \ref{relatedwork} is the related work and finally Section \ref{conclusion} concludes our work.

\section{Preliminaries}
\label{preliminaries}
\subsection{Local Differential Privacy (LDP)}
In local differential privacy setting, each user randomly perturbs its raw data and then sends the resulted data to the analyst. The untrusted analyst can only access the perturbed data other than the raw ones, which guarantees the privacy. Formally, let $\mathcal{T}$ denote the domain of a sensitive value, $\epsilon$-local differential privacy (or local privacy) is defined as follows.

\begin{Definition}
($\epsilon$-Local Differential Privacy, $\epsilon$-LDP). A randomized algorithm $\mathcal{A}$ satisfies $\epsilon$-local differential privacy (or local privacy), if and only if for (1) any pair of input $t_i,t_j \in \mathcal{T}$, and (2) any possible output $\mathcal{O}$ of $\mathcal{A}$, we have:\\
$$\frac{\mathbf{Pr}[\mathcal{A}(t_i)=\mathcal{O}]}{\mathbf{Pr}[\mathcal{A}(t_j)=\mathcal{O}]} \leq e^{\epsilon}.$$
\end{Definition}

Two vitally important properties of differential privacy are \emph{sequential composability} \cite{a9} and \emph{post-processing} \cite{post-processing}.

\begin{lemma}\label{sequential composability}
\textbf{(Sequential composability).} Given $m$ randomized algorithms $\mathcal{A}_i(1 \leq i \leq m)$, each of which satisfies $\epsilon_i$-LDP. Then the sequential application of $\mathcal{A}_i$ collectively provides $(\sum_{i=1}^{m} \epsilon_i)$-LDP.
\end{lemma}

\begin{lemma}\label{post processing}
\textbf{(Post-processing).} For any method $\phi$ which works on the output of an algorithm $\mathcal{A}$ that satisfies $\epsilon$-LDP without accessing the raw data, the procedure $\phi \big(\mathcal{A(\cdot)} \big)$ remains $\epsilon$-LDP.
\end{lemma}

\subsection{Frequency Oracles with LDP}
\label{fo protocol}
A frequency oracle (FO) protocol in the local setting enables the analyst to estimate frequency of any given value $x \in \mathcal{X}$ from all sanitized data received from users. In \cite{a8}, two effective protocols, generalized random response (GRR) and optimized local hash (OLH) were proposed to estimate frequencies with a large domain size $d = |\mathcal{X}|$.

\textbf{Generalized Random Response (GRR)}\cite{a8}: The GRR protocol makes each user answer correctly $y = x$ with probability $p = \frac{e^{\epsilon}}{e^{\epsilon} + d - 1}$, and answer wrongly $y \neq x$ with probability $q = \frac{1-p}{d-1} = \frac{1}{e^{\epsilon} + d - 1}$. Specially, it turns out that the fundamental randomized response (RR)\cite{rr} protocol is the special case when $d = 2$ and achieves the best performance\cite{a8,a12}. The shortage of GRR is that the estimated variance is linear with $d$:
\begin{equation}
\operatorname{Var}\big[\tilde{f}_\textit{grr}(t)\big] = n \cdot \frac{d - 2+e^{\epsilon}}{(e^{\epsilon}-1)^2} .
\label{grr variance}
\end{equation}

\textbf{Optimized Local Hashing (OLH)}\cite{a8}: The OLH protocol applies a hash function to map each input value into a value in $[g]$, where $g \geq 2$ and $g \ll d$. Then the GRR protocol is used to perturb the hash values in the domain $[g]$. In \cite{a8}, the optimal choice of the parameter $g$ is shown to be $\lceil e^{\epsilon}+1 \rceil$, which leads to the minimum variance.

Typically, let $\mathbb{H}$ be a universal hash function family, and $H$ be a function randomly chosen from $\mathbb{H}$ that outputs a value $x = H(v) \in [g]$ for every $v \in [d]$. The perturbing process is formalized as $\operatorname{Perturb}_{\textit{OLH}} \big(\langle H,x \rangle \big) = \langle H,y \rangle$, where
$$\forall_{i\in [g]} \mathbf{Pr} [y=i] =
\begin{cases}
p = \frac{e^{\epsilon}}{e^{\epsilon}+g-1},&\text{if $x=i$} \\
q = \frac{1}{e^{\epsilon}+g-1},&\text{if $x\neq i$}
\end{cases}.
$$

Then, the analyst counts the number of perturbed values that ``supports'' the input $t$, denoted by $\mathds{1}_t$, and transforms it to the unbiased estimation
\begin{equation}
\tilde{f}_\textit{olh}(t) = \frac{\mathds{1}_t - n/g}{p-1/g}.
\label{olh aggregate}
\end{equation}

Accordingly, the variance of this estimation is
\begin{equation}
\operatorname{Var}\big[\tilde{f}_\textit{olh}(t)\big] =n \cdot \frac{4e^{\epsilon}}{{(e^{\epsilon}-1)}^2} .
\label{olh variance}
\end{equation}


\textbf{Padding and Sampling based Frequency Oracle (PSFO)} \cite{a1,a2}: The PSFO protocol is used for mining frequent items over set-values of various lengths. The protocol can be described as a padding and sampling function $PS_l (\cdot)$, which specifies a maximum length $l$, and a frequency oracle (FO). Specifically, each user pads its transaction $t$ with dummy items to the specified maximum length $l$, and uses the FO protocol to transmit one item randomly sampled from the padded transaction. Then, the analyst applies the FO protocol to evaluate item frequencies. Finally, the estimated item frequency are corrected by multiplying the length $l$.

\subsection{Frequent-Pattern (FP) Tree}\label{sec:fptree}
FP-tree \cite{fp} is an efficient method for mining frequent patterns. We describe two aspects of it, namely, FP-tree construction and FP-growth mining.

\textbf{FP-tree Construction.} FP-tree is an extended prefix-tree structure for storing compressed, crucial information about frequent patterns. To construct a FP-tree, the transactions are first scanned and preprocessed, such that only frequent items are included and sorted in frequency descending order. The resulted transactions are then arranged in a tree structure, where if any two transactions share the same prefix, the shared part can be merged using one prefix structure with count fields accumulated properly. A header table links to patterns led by each frequent item. For example, Fig \ref{fptree} shows a FP-tree for five transactions.

\textbf{FP-Growth Mining.} FP-growth mines the complete set of frequent patterns in no privacy setting based on an FP-tree, without a costly candidate generation process. It starts from an initial frequent pattern of length 1, examines only a sub-database (called conditional pattern base) which consists of patterns composed of the initial pattern and the subsequent co-occurring frequent items, constructs a (conditional) FP-tree, and performs mining recursively with such a tree.
The FP-growth identifies long frequent patterns by searching through smaller conditional pattern base repeatedly. In this way, the cost of searching frequent patterns is substantially reduced.

\begin{figure*}[tb]
\centerline{\includegraphics[width=0.95\textwidth]{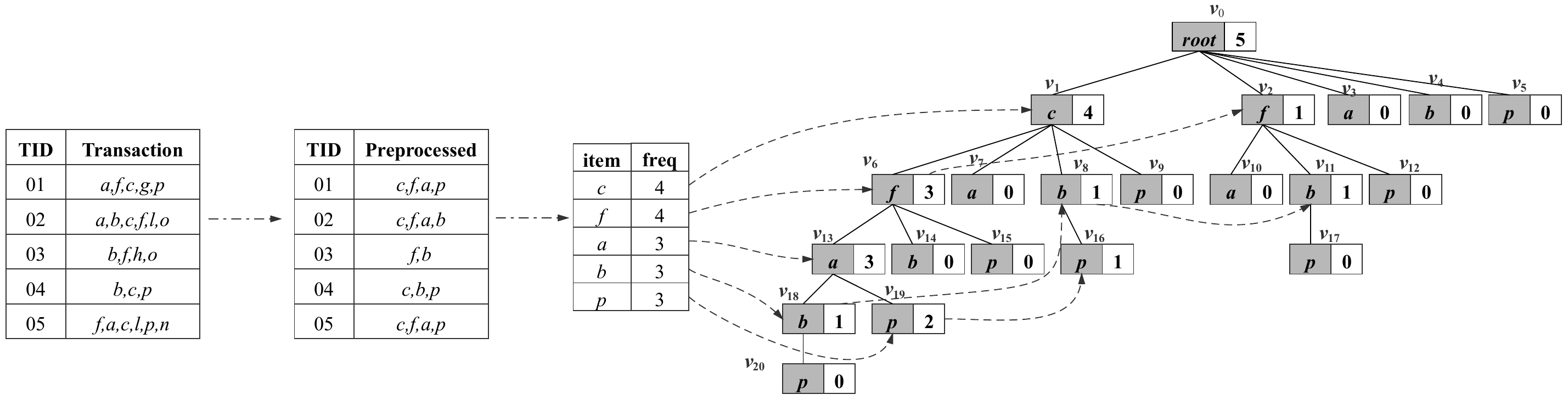}}
\caption{An example of frequent pattern tree (FP-tree).}
\label{fptree}
\end{figure*}

\section{Problem Overview }
\label{problem}
\subsection{Problem Definition}
In this paper, we focus on the task of discovering top-$k$ frequent itemsets over transactional (or set-valued) data in the context of LDP, where each user's input is a set of items, i.e., a transaction. This is  challenging due to the complicated transactional data as well as the exponentially growing potential itemsets. Formally, let $\mathcal{X} = \{x_1,x_2,...,x_d\}$ be the domain of $d$ distinct items. An itemset $X$ is a subset of $\mathcal{X}$, i.e., $X \subseteq \mathcal{X}$. Suppose there are $n$ users, the transaction of $i$th user is $t_i(i \in [1,n])$ and $\mathcal{T} = \langle t_1,t_2,...,t_n \rangle$ denotes the whole transactional database. The frequency of any itemset $X$ is the number of transactions containing $X$ in $\mathcal{T}$. That is,
$$f(X) = |\{ t_i | X \subseteq t_i , t_i \in \mathcal{T} \}|$$

Specifically, the untrusted analyst wants to discover the itemsets that occur most frequently. To control the size of output, it gives a minimum frequency threshold $\delta$ to output all itemsets whose frequency exceeds $\delta$ or a positive number $k$ to output the $k$ most frequent itemsets. In this paper, we focus on discovering the top-$k$ itemsets with the highest frequencies. Table \ref{notations} lists the main notations used in this paper.

\begin{table}[tb]
\caption{Notations.}\label{notations}
\begin{center}

\begin{tabu}{cl}\tabucline[1pt]{-}
  Symbol&Description \\\hline
  $\mathcal{T} = \langle t_1,t_2,...,t_n \rangle$ & the database of $n$ transactions \\
  $\mathcal{X} = \{x_1,x_2,...,x_d\}$ & the domain of $d$ distinct items\\
  $X$ & an itemset, $X \subseteq \mathcal{X}$ \\
  $f(\cdot),\tilde{f}(\cdot)$ & the actual and estimated frequency \\
  $\mathbb{G}(X)$ & the guessing frequency of $X$ \\
  $S^{\prime}$ &the frequent items set \\
  $\hat{\mathcal{N}}$ & the noisy FP-tree \\
  $M$ & the maximum level of the FP-tree\\
  $\dagger$ & the dummy value\\
  $\tilde{\mathcal{P}}$ & the set of frequent itemsets identified \\
\tabucline[1pt]{-}
\end{tabu}

\end{center}
\end{table}

\subsection{The SVSM Solution}\label{svim and svsm}
As far as we know, the state-of-the-art work to address the FIM task under LDP is the SVSM (Set-Value itemSet Mining) protocol\cite{a2}. Particularly, SVSM mined top-$k$ itemsets based on the $k$ most frequent items obtained by the SVIM protocol\cite{a2}. More details are as follows.

\textbf{Set-Value Item Mining (SVIM)}: The SVIM protocol focuses on discovering the $k$ most frequent items. It is in fact the PSFO protocol with the same problem setting as the LDPMiner\cite{a1}. SVIM divides all users into three mutually disjoint groups $G_1$, $G_2$ and $G_3$, and has three steps as follows.

\emph{Step 1. Prune domain $-\ G_1$.} Each user applies the OLH protocol to perturb one item randomly sampled from its input, i.e., $\operatorname{PS}_{l=1} (t)$. Then, the analyst estimates the frequency of each value in the original domain $\mathcal{X}$ and selects $2k$ items with the highest estimated frequencies as pruned domain $S$. The analyst broadcasts $S$ to all users, who prune their transactions by intersecting them with domain $S$.

\emph{Step 2. Size estimation $- \ G_2$.} Since each user possesses at most $L$ items when using padding and sampling (PS) protocol \cite{a1}, the selection of an appropriate $L$ is crucial. The basic strategy is to collect the length distribution of users and select a suitable $L$ in a private way. Specifically, each user $i$'s input is the number of items in its pruned transaction, i.e., $|t_i \cap S|$, and then given a fraction $\tau$, $L$ is computed as the smallest $l \in \{1,2,\ldots ,2k\}$ that satisfies
\begin{equation}
\frac{\sum_{j=1}^{l} \tilde{f}(j)}{\sum_{j=1}^{2k} \tilde{f}(j)} > \tau \label{length}
\end{equation}
For example, the 90th percentile length is the $l$ value with $\tau = 0.9$.

\emph{Step 3. Frequency estimation $- \ G_3$.} Once given $S$ and $L$, the PSFO protocol is applied to precisely estimate the frequencies of items in the pruned domain $S$. Firstly, each user $i$ inputs one item randomly sampled from its pruned transaction padded to length $L$, i.e., $PS_{l=L} (t_i \cap S)$. Then, the item frequencies are estimated according to PSFO, and the $k$ most frequent items are selected. 

SVIM cannot be used directly to mine itemsets, since the exponential growth of potential itemsets would incur to much noise in step 1. Therefore, SVSM protocol is proposed with the core technique named ``Guessing Frequency (GF)'' to construct a pony-size domain of itemsets which are likely to be frequent.

\textbf{Set-Value itemSet Mining (SVSM)}: Let $S^{\prime} = \{x_1,x_2,$ $...,x_k\}$ denote the set of top-$k$ items returned by SVIM. Then, the guessing frequency of each potential itemset $X \subseteq S^{\prime}$ is computed by \eqref{gf}, and the pony-size domain set $\textit{IS}$ of itemsets is constructed by selecting $2k$ itemsets with highest guessing frequencies.

\begin{Definition}
(Guessing Frequency, GF). Let an itemset $X$ be a subset of a set of known items $I = \{{x_1,x_2,...,x_m} \}$, i.e. $X \subseteq I$. The frequency of the $i$th item $x_i$ is denoted by $f(x_i)$. The guessing frequency of $X$ is $\mathbb{G}(X)$, defined as follows:\\
\begin{equation}\label{eq:guessfreq}
\mathbb{G}(X) = \prod_{x_i \in X} \frac{\gamma \times f(x_i)}{\max (f)} .
\end{equation}
where $\max (f)$ is the maximum frequency of all items, and $0 \leq \gamma \leq 1$ is a predifined parameter.
\label{guess fre}
\end{Definition}

Once the domain of potential itemsets is pruned, SVSM uses the subsequential steps of SVIM to identify frequent itemsets. The major difference is that the input of each user $i$ is a set of itemsets contained by both $i$'s input transaction and $\textit{IS}$, that is, $\textit{tx}_i = \{ X | X\in \textit{IS}, X\subseteq t_i\}$.

\section{LDP-FPMiner}
\label{fpmine sum}
In this section, we present the LDP-FPMiner that discovers the $k$ most frequent itemsets under LDP. The main idea is to construct a noisy FP-tree, which allows the untrusted analyst to mine itemsets privately. In the following, Section \ref{fpmine} overviews the scheme. Section \ref{construct and mine} and Sections \ref{cutdown candidate} describe the details of constructing a noisy FP-tree. Section \ref{sec:mining} outlines the mining of the noisy FP-tree, and \ref{sec:analysis} provides theoretical analysis of the scheme.

\subsection{Overview}
\label{fpmine}

The overview of LDP-FPMiner is depicted by Figure~\ref{fig:overview}. It has three steps: First, a set of frequent items is identified with the SVIM protocol. Second, OLH protocol is applied to approximate the maximum number of frequent items that users hold, that is the height of the FP-tree. Third, it constructs a noisy FP-tree with LDP in breadth-first order, and mines it for frequent itemsets by the FP-growth algorithm \cite{fp}. The overall procedure is presented in Algorithm \ref{alg fpmine}.

\begin{figure}[tb]
\centerline{\includegraphics[width=0.45\textwidth]{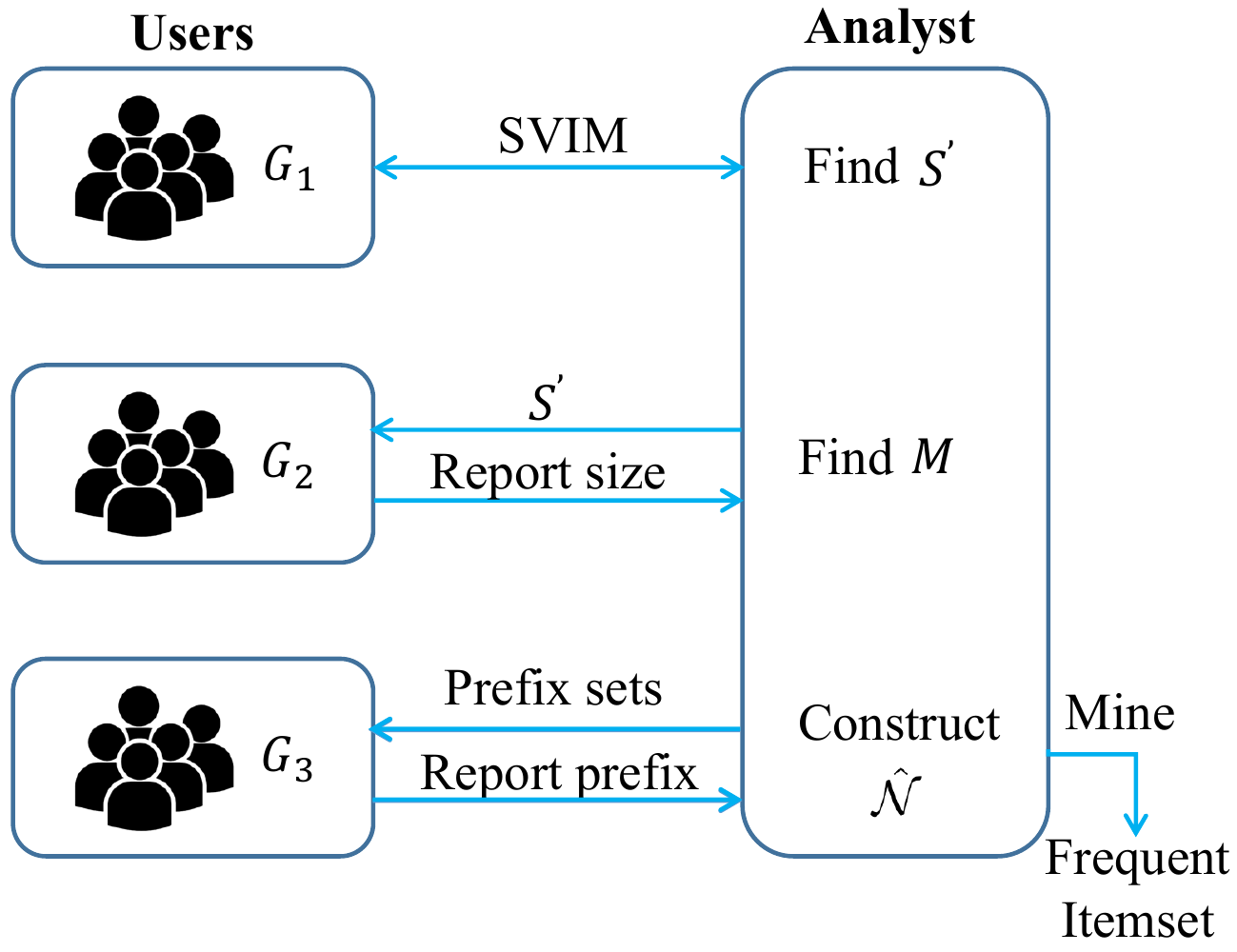}}
\caption{The Overview of LDP-FPMiner.}
\label{fig:overview}
\end{figure}

\begin{algorithm}[hb]
  \caption{LDP-FPMiner($\mathcal{T},\mathcal{X},k,\epsilon$)}
  \label{alg fpmine}
  \begin{algorithmic}[1]
  \STATE Randomly divide $\mathcal{T}$ into three groups $G_1,G_2,G_3$; \label{fpmine group}

  \textbf{Find Frequent Items}:
  \STATE Collect the items set $S^{\prime} \gets \operatorname{SVIM}(G_1,k,\epsilon)$; \label{fpmine items}

  \textbf{Find Tree Height:}
  \STATE Each user in $G_2$ perturbs the number of frequent items it holds with $\operatorname{OLH}({\epsilon})$;\label{fpmine length first}
  \STATE Compute the 80th percentile length $M$ by \eqref{length};\label{fpmine Lm}

  \textbf{Find Frequent Itemsets:}
  \STATE $\hat{\mathcal{N}} \gets \operatorname{ConstructNoisyTree}(G_3,S^{\prime},M,\epsilon)$;
  \STATE Mine $\hat{\mathcal{N}}$ and release the $k$ itemsets set $\tilde{\mathcal{P}}$; \label{fpmine mine}
  \RETURN $\tilde{\mathcal{P}}$
  \end{algorithmic}
\end{algorithm}

Note that, in the first step, SVIM protocol (line \ref{fpmine items} of Algoritm~\ref{alg fpmine}) is used to find the frequent items and their frequencies. The former are used to find the frequent itemsets in later two steps, while the latter are used to produce guessing frequencies for optimizations. In the second step, we set $M$ empirically as the 80th percentile of the number of frequent items that user holds.

As described in Section \ref{svim and svsm}, the SVSM protocol uses the PSFO protocol as a building block and constructs exponentially growing candidate itemsets. In this section, by leveraging the FP-tree approach, we aim to mine frequent itemsets under LDP without a costly candidate generation process. Moreover, by making use of the FP-tree structure, we want to reduce the scale of noise added, and efficiently identify frequent itemsets with a high accuracy. However, there are three challenges as follows.

The first challenge is that in LDP setting the analyst must construct an FP-tree over sanitized data. There is no raw user data available to calculate accurate frequencies. Particularly, the original FP-growth \cite{fp} constructs a non-privacy FP-tree by scanning raw transactions and updating nodes in the depth-first order. But this cannot be implemented under LDP since the frequencies of patterns are not directly available. To overcome this, we construct a noisy FP-tree in the breadth-first order based on the fundamental FO protocol (e.g. OLH). Given a maximum level $M$ collected privately (explained in Section \ref{construct and mine}), our construction algorithm queries the users about pattern frequencies for each level, and then constructs the noisy FP-tree level by level.

The second challenge is that the number of patterns generated at an FP-tree level is exponentially explosive. Querying the users about such explosive number of pattern frequencies will soon become infeasible, and severely degrade the accuracy due to a large amount of noise injected. To address this issue, we cut down the number of candidate patterns at each level to no more than $\xi k$, and expand only these patterns in the FP-tree.
Therefore, we design the function $\operatorname{CutDownCandidate}$ (Algorithm \ref{alg:cutdown candidate} in Section \ref{cutdown candidate}) to construct a pony-size candidate set of nodes that is likely to be frequent.

Finally, how to optimize a noisy FP-tree based on the tree structure is challenging. There are frequency consistency constraints between FP-tree nodes that can be used. For example, the sum of the counts of children nodes should be equal to or less than the count of the parent node. However, since we cut down a part of nodes during construction and thus the resulted FP-tree is not complete, whether these constraints can be used should be carefully examined. Additionally, we also optimize an noisy FP-tree by using guessing frequencies to adjust the frequency evaluation in FP-tree construction and mining.

\subsection{Constructing a Noisy FP-tree}
\label{construct and mine}

In the non-privacy setting, FP-trees are originally constructed in depth-first order \cite{fp}. In the LDP setting, since pattern frequencies are not directly available and should be queried privately, the depth-first order requires too many frequency queries, and would quickly consume the privacy budget. Lee et al.\cite{a5} also built noisy FP-trees to derive itemset frequencies in the centralized differential privacy (CDP) setting, where raw data is accessible. To the best of our knowledge, we are the first to propose an approach to constructing a FP-tree with LDP.

Specifically, we construct locally differentially private FP-trees in the breadth-first order. Let $M$ be the maximum level of the tree, and $x_1 \succ x_2 \succ \cdots  \succ x_k$ be $k$ frequent items in frequency descending order, which consist a set $S^{\prime}$. Our construction method has the following two phases:

\textbf{Phase \uppercase\expandafter{\romannumeral1}: Preprocessing.} All users prune their transactions, remain only the frequent items and rearrange them in the frequency descending order. For instance, Fig. \ref{fptree} shows the five preprocessed transactions when the frequent items are rearranged into the frequency descending order $c\succ f\succ a\succ b\succ p$.

Notably, after preprocessing, massive non-frequent items are pruned. This significantly cuts down the number of candidate itemsets, and thus improves the performance. The underlying basis is the Apriori property\cite{apr}: \emph{only if the length-$\alpha$ itemset is frequent are its length-($\alpha +1$) supersets likely to be frequent}.

\textbf{Phase \uppercase\expandafter{\romannumeral2}: Constructing noisy FP-tree.} We construct a noisy FP-tree with LDP in breadth-first order as described in Algorithm \ref{alg ConstructNoisyTree}.

First, initialization is done as follows. (1) All users are divided into $M$ equal groups (line \ref{ConstructNoisyTree group}), each of which is used for frequency query at a level of the FP-tree. (2) Each user prunes her transaction in term of the set of frequent items $S^{\prime}$, and sorts the frequent items properly for the later FP-tree construction (line \ref{ConstructNoisyTree preprocess}). (3) The noisy FP-tree $\hat{\mathcal{N}}$ is initialized with a valid root holding a count $n_g$ (line \ref{ConstructNoisyTree tree}), and the root can be regarded the $0$-th level of the tree having a prefix candidate set $C_0$.

Then, the noisy FP-tree is constructed in the breadth-first order as below. (1) For each level $l$, the nodes at this level are added and their corresponding prefixes are constructed as the candidate set $C_l$ (line \ref{ConstructNoisyTree prefix set} - \ref{ConstructNoisyTree prefix set constructed}). Subsequently, the $\operatorname{CutdownCandidate}$ (line \ref{ConstructNoisyTree cutdown}) is invoked to cut down $C_l$ into small size $\xi \cdot k$ (explained in Section \ref{cutdown candidate}). (2) For each level $l$, the users in the corresponding group perturb their inputs in term of the prefix candidate set $C_l$ with the OLH protocol, and the analyst computes an estimate count $\tilde{f}(v)$ for each node $v$ (line \ref{ConstructNoisyTree estimate}-\ref{ConstructNoisyTree collect}). The nodes with negative estimated counts are updated with 0 counts (line \ref{ConstructNoisyTree update}). The algorithm repeats level by level until it reaches the maximum level $M$. Finally, the algorithm returns the noisy FP-tree $\hat{\mathcal{N}}$.

In the following, we remark on the main points of the noisy FP-tree construction algorithm $\operatorname{ConstructNoisyTree}$.

\begin{itemize}
\item A node $v$ in the noisy FP-tree has two fields: $v.item$ and $v.count$, where $v.item$ denotes the indicated item of node $v$ and $v.count$ represents the count of times its prefix $\bar{p}_v$ appears in database, respectively. For example, in Fig. \ref{fptree}, the node $v_6$ indicates the item $f$ (shown in shaded box) and $v_6.count = 3$ means the prefix $(c,f)$ appears 3 times, i.e., users T01,T02 and T05 includes this prefix in their transactions.

\item Both $S^{\prime}$ and $M$ are obtained with $\epsilon$-LDP. Meanwhile, since we construct $\hat{\mathcal{N}}$ in breadth-first order, each level of the tree is completely dependent on the previous level (line \ref{ConstructNoisyTree prefix set}-\ref{ConstructNoisyTree prefix set constructed}), which is collected with LDP (line \ref{ConstructNoisyTree estimate}-\ref{ConstructNoisyTree collect}). Therefore, the noisy FP-tree $\hat{\mathcal{N}}$ does not disclose any privacy of the specific transaction.

\item The input of each user at level $l$ ($1\leq l \leq M$) (line \ref{ConstructNoisyTree perturbe}) is a prefix in $C_l$ or the dummy value $\dagger$ (if her first $l$ items does not exist in $C_l$). We apply OLH protocol with the finite domain $C_l \cup \dagger$ to gather information. For example, when $l=3$ and $C_3 = \{\bar{p}_{v_{13}}, \bar{p}_{v_{14}}, \bar{p}_{v_{15}}, \bar{p}_{v_{16}}, \bar{p}_{v_{17}}\}$ in Fig. \ref{fptree}, the input of user T01 is $\bar{p}_{v_{13}} = (c,f,a)$, which is her first three preprocessed items, while that of T03 is the dummy value $\dagger$.

\item We cut down the size of the domain set (line \ref{ConstructNoisyTree cutdown}) as well as filter out the nodes with negative counts (line \ref{ConstructNoisyTree update}) on each iteration to improve accuracy, which will be explained in Section \ref{cutdown candidate}.

\item We randomly divide users into $M$ equal-sized groups and users in each group use the full privacy budget $\epsilon$. Meanwhile, the estimated count of node $v$ should multiply $M$ to correct the underestimation. It has turned out that the overall process achieves better accuracy and satisfies $\epsilon$-LDP as well.
\end{itemize}

\begin{algorithm}[t]
  \caption{ConstructNoisyTree($G_3,S^{\prime},M,\epsilon$)}
  \label{alg ConstructNoisyTree}
  \begin{algorithmic}[1]

  \STATE \textbf{// Initialize:}
  \STATE Randomly divide users into $M$ groups $g_1,g_2,...,g_{M}$ of the same size $n_g = \lfloor \frac{|G_3|}{M} \rfloor$; \label{ConstructNoisyTree group}
  \STATE Each user prune her items not in $S^{\prime}$ and rearrange left frequent items in frequency descending order;\label{ConstructNoisyTree preprocess}
  \STATE Initialize tree $\hat{\mathcal{N}}$ with a root $v_r$, and set $v_r.count = n_g$;\label{ConstructNoisyTree tree}
  \STATE Mark $v_r$ as valid, and set $C_0 = \{v_r\}$;

  \FOR{$l=1$ to $M$} \label{ConstructNoisyTree begin loop}

    \STATE \textbf{// Generate Candidates:}
    \WHILE{there is a valid $v \in C_{l-1}$ and $v.count > 0$} \label{ConstructNoisyTree prefix set}
      \STATE Initialize a candidate prefix set $C_l = \emptyset$;
      \STATE Mark $v$ as invalid;
      \FOR{each item $x \in S^{\prime}$ and $v.item \succ x.item $} \label{ConstructNoisyTree add child}
        \STATE Add a child $v_c$ of $v$ with item $x$ and count 0;\label{ConstructNoisyTree child}
        \STATE Mark $v_c$ as valid and obtain its prefix $\bar{p}_{v_c}$;\label{ConstructNoisyTree child prefix}
        \STATE $C_l \gets C_l \cup \bar{p}_{v_c}$;\label{ConstructNoisyTree child prefix set}
      \ENDFOR
    \ENDWHILE \label{ConstructNoisyTree prefix set constructed}
	\STATE $C_l \gets \operatorname{CutdownCandidate}(S^{\prime},C_l,\xi)$;\label{ConstructNoisyTree cutdown}

    \STATE \textbf{// Query Frequencies:}
    \FOR{each user in group $g_l$} \label{ConstructNoisyTree estimate}
    \STATE Perturbe her input ( i.e., first $l$ items) with OLH; \label{ConstructNoisyTree perturbe}
    \ENDFOR \label{ConstructNoisyTree ldp}

    \STATE Collect the estimated count $\tilde{f}(v)$ of each node $v$ at level $l$ using the domain $C_{l} \cup \dagger$; \label{ConstructNoisyTree collect}

    \STATE Update nodes by coverting all negative counts to 0;\label{ConstructNoisyTree update}

  \ENDFOR \label{ConstructNoisyTree end loop}

  \RETURN The noisy FP-tree $\hat{\mathcal{N}}$.

  \end{algorithmic}
\end{algorithm}

\subsection{Cutting Down Cadidate Prefix Set}
\label{cutdown candidate}
Recall that, during the noisy tree construction, the size of initial candidate prefixes set $C_l$ at each level $l$ soon becomes very large (e.g. thousands or more). This would degrade the accuracy greatly. According to Algorithm~\ref{alg ConstructNoisyTree}, this size is expanded once at each level, with a maximal factor $k$. Our idea is to cut down the size at each level by pruning candidate prefixes when the size exceeds a certain value (i.e., $\xi k$). Although this cutdown may cause information loss and lead to underestimation, it overcomes the size expansion issue and works experimentally well.

Another reason to cutdown candidate prefixes is that there are many redundant prefixes. For example, at the level-3 in Fig. \ref{fptree}, the set $C_3$ is initially $\{\bar{p}_{v_{6}}, \bar{p}_{v_{7}}, \bar{p}_{v_{8}}, \bar{p}_{v_{9}}, \bar{p}_{v_{10}}, \bar{p}_{v_{11}}, \bar{p}_{v_{12}}\}$. However, as the counts are collected, half of the prefixes (i.e., $\bar{p}_{v_{7}}, \bar{p}_{v_{9}},\bar{p}_{v_{10}}, \bar{p}_{v_{12}}$) should be pruned. Thus, if we can prune these meaningless nodes in advance, then we can effectively reduce the noise added.

Specifically, inspired by SVSM\cite{a2}, we limit the candidate prefix set within a fixed size $\xi \cdot k$ in term of temporal guessing frequencies, as shown in Algorithm \ref{alg:cutdown candidate}. Here, $\xi$ is an adjustable parameter. The temporal guessing frequency $\mathbb{T} (\bar{p}_v)$ of each candidate prefix $\bar{p}_v \in C$ is computed as follows.
\begin{equation}
\mathbb{T} (\bar{p}_v) = \tilde{f}(\bar{p}_{\operatorname{parent}(v)}) \cdot \bar{f}(v)
\end{equation}
where $\tilde{f}(\bar{p}_{\operatorname{parent}(v)})$ denotes the estimated frequency of the prefix of $v$'s parent node, and $\bar{f}(v)$ is a probability computed  based on the normalized estimated frequencies of frequent items. Assuming that the items of node $v$ and its parent node are the $(i+j)$th and $i$th frequent items $x_{i+j}$ and $x_i$, respectively, then $\bar{f}(v)$ can be computed by Eq.~\eqref{eq:xi_probability}.

\begin{equation}\label{eq:xi_probability}
\bar{f}(v) =\tilde{f}(x_{i+j}) \prod_{t=1}^{t=j-1}(1-\tilde{f}(x_{i+t}))
\end{equation}

Note that the computation of Eq.~\eqref{eq:xi_probability} is due to the FP-tree structure. A node $v$ appears as a child of its parent only if all frequent items ranked between its parent node and itself do not appear.

Then the $\xi k$ candidate prefixes with highest temporal guessing frequencies are selected to form the pony-size set $C^{\prime}$. The intuition is that a prefix with a high estimated frequency is more likely to be split with frequent items. Note that the temporal frequency of each candidate prefix depends only on the frequencies queried previously and can be computed within $O(1)$ time.

\begin{algorithm}[t]
\caption{CutdownCandidate($S^{\prime},C,\xi$)}
\label{alg:cutdown candidate}
\begin{algorithmic}[1]

\STATE Initialize $C^{\prime}$;
\IF{$|C| > \xi \cdot k$}
  \FOR{each candidate prefix $\bar{p}_v \in C$}
    \STATE Compute the temporal guessing frequency $\mathbb{T} (\bar{p}_v)$;
  \ENDFOR
  \STATE Construct $C^{\prime}$ by selecting the $\xi k$ prefixes with highest temporal guessing frequencies;
\ELSE
  \STATE $C^{\prime} \gets C$

\ENDIF

\RETURN $C^{\prime}$
\end{algorithmic}
\end{algorithm}

\subsection{Mining a Noisy FP-tree}
\label{sec:mining}

So far, we have obtained a noisy FP-tree satisfying LDP. The final step for this scheme is to mine this noisy FP-tree for frequent itemsets. This mining procedure is completely a post-process, and it is the same as that of no-privacy setting. Namely, it follows the original FP-growth algorithm, which has been outlined in Section~\ref{sec:fptree}, and more details can be found in \cite{fp}. 

\subsection{Theoretical Analysis}
\label{sec:analysis}

\subsubsection{Computational Complexity}

The computational complexity of constructing a noisy FP-tree is given in Theorem~\ref{fp computational complexity}.

\begin{theorem}
The computational complexity of constructing a noisy FP-tree for LDP-FPMiner is $O(k^3)$.\label{fp computational complexity}
\end{theorem}

\begin{proof}
The computational complexity is dominated by the main loop in Algorithm \ref{alg ConstructNoisyTree} (line \ref{ConstructNoisyTree begin loop}-\ref{ConstructNoisyTree end loop}), which terminates in $M$ iterations. Since there are $k$ frequent items, the length of preprocessed transactions is not greater than $k$, and so we have $M \leq k$. Furthermore, for each iteration, one candidate prefix set is constructed (line \ref{ConstructNoisyTree prefix set}-\ref{ConstructNoisyTree cutdown}) as the domain to further gather estimations (line \ref{ConstructNoisyTree estimate}-\ref{ConstructNoisyTree ldp}). Since the set is then pruned into size $O(k)$, the number of nodes reserved is $O(k)$ at each level and may generate $O(k(k-1)) = O(k^2)$ children nodes. Therefore, the computational complexity is $O(M k^2) = O(k^3)$.
\end{proof}

\subsubsection{Estimation Accuracy}

The estimation accuracy of a prefix frequency is given in Theorem~\ref{the:accuracy}.

\begin{theorem}\label{the:accuracy}
For any node $v$ in the FP-tree, let $\tilde{f}(v)$ be the estimated count of its prefix $\bar{p}_{v_c}$ collected privately in Algorithm \ref{alg ConstructNoisyTree}, and $f(v)$ be the actual count. When the number of users participating at each level is $n_g$, and there are $M$ levels (excluding the 0-level, i.e., root node ), with at least $1-\beta$ probability, we have
\begin{equation}
\max{|\tilde{f}(v) - f(v)|} = O \left(\frac{ \sqrt{n_g} \sqrt{ \log (1/\beta)}}{\epsilon} \right).
\end{equation}
\end{theorem}

\begin{proof}
according to \eqref{olh aggregate},
\begin{equation}\nonumber
	\begin{aligned}
\left |\tilde{f}(v) - f(v) \right| & = \left| \sum_{j=1}^{n_g} \frac{\tilde{v}^j - 1 / g}{p - 1/g} - \sum_{j=1}^{n_g}\frac{v^j - 1 /g}{p - 1/g} \right| \\
                      & = \left| \sum_{j=1}^{n_g}  \frac{\tilde{v}^j - v^j}{p - 1/g} \right| .\\
	\end{aligned}
\end{equation}
where $g = e^{\epsilon}+1$ and $p = \frac{e^{\epsilon}}{e^{\epsilon} + g -1} = 1/2$ are the optimal probability setting in \cite{a8}. Random variable $\tilde{v}^j$ and $v^j$ denote the estimated and actual count of the $j$th user, respectively.

Accordingly, we have
$$Var\left[\tilde{v}^j - v^j \right] = O \left(\frac{1}{\epsilon^2} \right),$$
$$ \left| \frac{\tilde{v}^j - v^j}{p - 1/g} \right| \leq \left| \frac{1}{p - 1/g} \right| = \frac{2(e^\epsilon + 1)}{e^\epsilon - 1}.$$

By Bernstein's inequality,
\begin{equation}\nonumber
\begin{aligned}
&Pr\left[\left |\tilde{f}(v) - f(v) \right| \geq  \lambda \right] \\
& =  Pr \left[\left| \sum_{j=1}^{n_g} \frac{\tilde{v}^j - v^j}{p - 1/g} \right|  \geq \lambda \right]\\
& \leq 2 \times exp \left \{- \frac{ \frac{\lambda^2}{ 2}}{\sum_{j=1}^{n_g}Var \left[ \frac{\tilde{v}^j - v^j}{p - 1/g} \right] + \frac{\lambda}{3}\cdot \frac{2(e^\epsilon + 1)}{e^\epsilon - 1} } \right\} \\
& = 2 \times exp \left \{- \frac{\lambda^2 / 2}{n_g \cdot O\left(\frac{1}{\epsilon^2}\right) + \lambda \cdot O\left(\frac{1}{\epsilon} \right)}  \right\}.
\end{aligned}
\end{equation}

Therefore, $\left |\tilde{f}(v) - f(v) \right| <  \lambda$ holds with at least $1-\beta$ probability while $\lambda = O\left(\frac{\sqrt{n_g} \cdot \sqrt{\log (1/\beta)}}{\epsilon} \right)$.
\end{proof}

\subsubsection{Local Differential Privacy}

The local differential privacy of LDP-FPMiner is given in Theorem~\ref{the:lpd-fpminer} and Theorem~\ref{the:ldp-construct}.

\begin{theorem}
The LDP-FPMiner (i.e., Algorithm \ref{alg fpmine}) satisfies $\epsilon$-LDP.\label{the:lpd-fpminer}
\end{theorem}

\begin{proof}
LDP-FPMiner divides all users into three mutually disjoint groups, that is, group $G_1$ for finding frequent items $S^{\prime}$, $G_2$ for computing maximum length $M$ and $G_3$ for constructing a noisy FP-tree $\hat{\mathcal{N}}$. LDP-FPMiner provides $\epsilon$-LDP protection for the first two groups of users obviously, and for the third group due to Theorem~\ref{the:ldp-construct}. The mining process over the noisy FP-tree does not consume any privacy budget due to the post-processing property. Thus the whole process of LDP-FPMiner satisfies $\epsilon$-LDP.
\end{proof}

\begin{theorem}
Constructing a noisy FP-tree (i.e., Algorithm \ref{alg ConstructNoisyTree}) satisfies $\epsilon$-LDP.\label{the:ldp-construct}
\end{theorem}

\begin{proof}
Let $V_l$ denote the set of nodes at the $l$th level in a noisy FP-tree $\hat{\mathcal{N}}$. Then there are at most $M$ iterations to privately construct $\hat{\mathcal{N}}$ (excluding the root node). For each iteration, algorithm \ref{alg ConstructNoisyTree} allocates $n_g$ users to gather information of $V_l$. Due to the fact that each user applies OLH to perturb her input only once on the iteration participated with full privacy budget $\epsilon$, thus every user is protected by $\epsilon$-LDP and the overall process of constructing a noisy FP-tree satisfies $\epsilon$-LDP.
\end{proof}

\section{Optimizations}
\label{optimizations}
We optimize LDP-FPMiner using frequency consistency between estimated frequencies and guessing frequencies, and also using frequency consistency in term of the FP-tree structure. These optimizations satisfy the post-processing property of local differential privacy (Lemma \ref{post processing}), and thus the local differential privacy is cerntainly preserved.

\subsection{Guessing Probability}

We use the estimated frequencies of frequent items to compute guessing probabilities for all candidate prefixes in a noisy FP-tree. As previous, let $x_1 \succ x_2 \succ \cdots  \succ x_k$ be $k$ frequent items in frequency descending order. For a candidate prefix $p = x_{i_1}x_{i_2} \cdots x_{i_a}$, with $1 \le i_1 < i_2 < \cdots < i_a \le k$, we can compute its guessing probability by Eq.~\eqref{eq:guess_probability}.

\begin{equation}\label{eq:guess_probability}
\bar{f}(p) = \tilde{f}(x_{i_a}) \prod_{u=1}^{u=a-1} \Big [ \tilde{f}(x_{i_u}) \cdot \prod_{v=i_u+1}^{v=i_{u+1}-1}(1-\tilde{f}(x_{v}))\Big ]
\end{equation}

The guessing probability of a prefix can be regarded as the estimated occurrence probability of the prefix. For example, in Figure~\ref{fig:guessprob}, we can compute $\bar{f}(x_1) = \tilde{f}(x_1)$, $\bar{f}(x_2) = (1-\tilde{f}(x_1))\tilde{f}(x_2)$, $\bar{f}(x_3) = (1-\tilde{f}(x_1))(1-\tilde{f}(x_2))\tilde{f}(x_3)$, $\bar{f}(x_1 x_2)$ $ = \tilde{f}(x_1) \tilde{f}(x_2)$, and $\bar{f}(x_1 x_3) = \tilde{f}(x_1)(1- \tilde{f}(x_2))\tilde{f}(x_3)$. The guessing probabilities can be used to generate guessing frequencies, and then used to optimize noisy FP-trees by exploiting frequency consistency that both queried and guessing frequencies should be consistent.

\begin{figure}[tb]
\centerline{\includegraphics[width=0.40\textwidth]{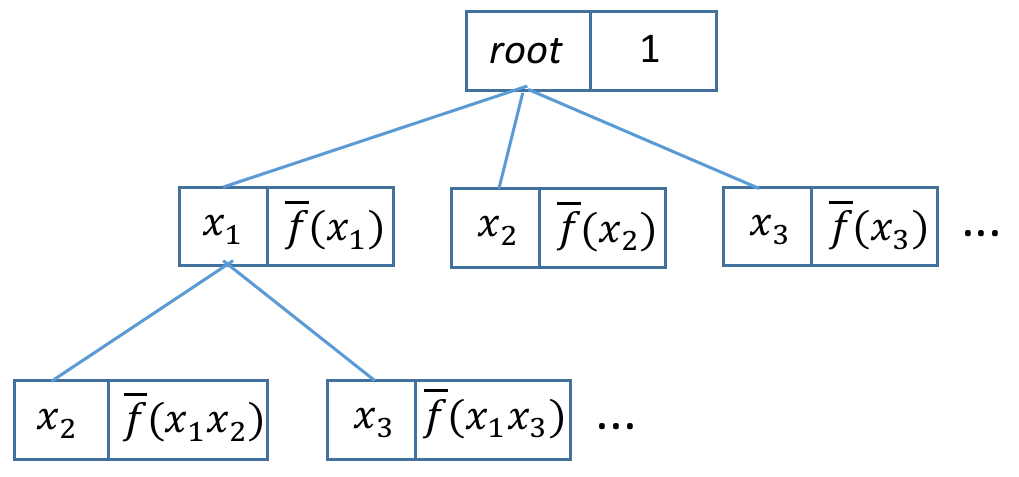}}
\caption{An Example of Guessing Probability.}
\label{fig:guessprob}
\end{figure}

\subsection{Two-folded Weighted Combination}

Since the space of itemsets is large, most of the itemsets have low frequencies. When these frequencies are queried with LDP and their statistics are made with partial data, they may deviate obviously from their original values, leading to biased evaluations. The noisy FP-tree based algorithms suffer from this issue due to the candidate cutdown. To address this, we use the guessing frequencies derived from the guessing probabilities to alleviate the bias, and our solution is two-folded.

First, during the noisy FP-tree construction, we rescale the guessing probabilites generated by Eq.~\eqref{eq:guess_probability} with their maximum to get guessing frequencies, and  for each level, we combine these guessing frequencies with the queried frequencies which are also rescaled with their maximum. Let $\mathbb{P}$ denote the set of all prefixes, then for any prefix $p \in \mathbb{P}$, and its guessing frequency $\mathcal{G}(p)$, we have the weighted frequency $\Omega(p)=\omega^{\prime} \cdot \tilde{f}(p) + (1-\omega^{\prime})\cdot \mathcal{G}(p)$, where $\omega^{\prime} \in [0,1]$ is the prefix weighted parameter. This operation, which we call prefix weighted combination (PWC), alleviates the bias of the queried frequencies effectively due to the experimental results given in Section~\ref{experiment}.

Second, in the final itemset mining, the guessing frequency of each itemset is computed by Eq.~\eqref{eq:guessfreq} and combined with the mined frequencies. Note that the computation of guessing frequencies of itemsets are quite different from those of prefixes, since an itemset usually can be found in multiple prefixes. Let $\tilde{P}$ denote the set of itemsets, which consists of the discovered itemsets. Then for any itemset $X \in \tilde{P}$, with its guessing frequency $\mathcal{G}(X)$, we have the weighted frequency $\Omega(X)=\omega \cdot \tilde{f}(X) + (1-\omega)\cdot \mathcal{G}(X)$, where $\omega \in [0,1]$ is the itemset weighted parameter. As can be seen from the experiments, this operation named itemset weighted combination (IWC) greatly improves the performance of our scheme.


\subsection{Conditional Constrained Inference}

 We optimize a noisy FP-tree by the tree structure constraint that \emph{the count of a node is no less than the sum of counts of all its children}. Noisy counts may break this constraint. The challenge is that when equation holds is unknown, and even worse, we cut down nodes during the noisy FP-tree construction and finally only get a partial tree. To overcome this, we regard each node and all its children as a 2-layer tree, and optimize each 2-layer noisy tree proportionally to the guessing probabilities of the remaining children nodes (some children nodes may be cut down). Specifically, given a node $v$ in $\hat{\mathcal{N}}$, denoting $\tilde{f}(.)$ as the noisy count function, $\operatorname{succ}(v)$ as the set of all children of node $v$, and $b$ as the number of children, we correct the noisy counts by constrained inference \cite{hay2010boosting,qardaji2013understanding} with apropriate ratio adjustments by Eqs.~\eqref{eq:fstar_v} and \eqref{eq:fstar_vc}.
\begin{equation}\label{eq:fstar_v}
\tilde{f^*}(v) = \frac{b^2 - b}{b^2 - 1} \tilde{f}(v) + \frac{b-1}{b^2-1} \cdot \frac{1}{\theta} \cdot \sum_{v_c \in \operatorname{succ}(v)} \tilde{f}(v_c)
\end{equation}
\begin{equation}\label{eq:fstar_vc}
\tilde{f^*}(v_c) = \tilde{f}(v_c) + \frac{1}{b}(\tilde{f^*}(v) \cdot \theta - \sum_{v_c \in \operatorname{succ}(v)} \tilde{f}(v_c))
\end{equation}

Here, $\theta$ is the ratio of the sum of guessing probabilities of the remaining children nodes to the gessing probability of the parent node. $\theta$ is computed by Eq.~\eqref{eq:alpha}.
\begin{equation}\label{eq:alpha}
\theta = \frac{1}{\bar{f}(p(v))} \sum_{v_c \in \operatorname{succ}(v)} \bar{f}(p(v_c))
\end{equation}
where $p(v)$, $p(v_c)$ represent the prefixes of the parent node $v$ and the children nodes $v_c$. Note that, when 
\[\tilde{f}(v) < \sum_{v_c \in \operatorname{succ}(v)} \tilde{f}(v_c),\]
we simply set $\theta = 1$ due to the probable large amount of noise in children nodes. Since when $\theta$ is small, the noise may impact the correcton results greatly, in the experiments, we set a threshold value $\theta_0$ (e.g,  $\theta_0 = 0.3$) and do constrained inferences conditionally only when $\theta \ge \theta_0$. This process can be repeated several times to get a better result. The times of repetion can be determined experimentally, and in our experiments, we find 5-10 times of repetion seem sufficient.



\subsection{Negative-positive Balance}

When querying prefix frequencies from users in LDP setting during the noisy FP-tree construction, the estimated frequencies may be positive or negative. We propose a post-processing method called negative-positive balance (NPB) to reduce the noise added effectively. 

The NPB method works as follows. First, the analyst calculates the sum of absolute values of all negative frequencies, and set all negative frequencies to 0; Then, the analyst randomly picks a positive frequency and substracts 1 from it, and this repeats until the value substracted from positive frequencies is equal to the absolute sum of negative frequencies. 

The intuition of this method is that frequencies are originally non-negative, and if an estimated frequency is negative, it is certainly underestimated. Thus, setting negative frequencies to 0 is on the right way to reduce the noise. However, since we add some value to the overall frequencies, and we need to substract the same value from them to maintain unbiasedness. Moreover, substracting values from positive frequencies makes about half of them reduce the noise, since roughly half of the positive frequencies should be overestimated in probability. Theorefore, most of frequencies are on their right way to reduce the noise, and the NP balance would reduce the noise overall.

\section{Experiments}\label{experiment}
In this section, we experimentally evaluate the performance of LDP-FPMiner, and compare it with the state-of-the-art protocol SVSM. All experiments are performed on an Intel Core i5-7500 3.4GHz CPU with 16GB RAM.

\subsection{Settings}
\label{setting}
We implement both LDP-FPMiner and SVSM in Python 3.8. For both schemes, we partition all users into three groups in the same way. Specifically, $50\%$ of users report in the first step to identify $k$ frequent items as well as their frequencies, $10\%$ of users report size, and $40\%$ of users participate in constructing the noisy tree for LDP-FPMiner, and evaluating frequencies of candidate itemsets for SVSM. 

\begin{table}[tb]
\caption{Dataset description. The numbers of transactions $|\mathcal{T}|$ and the dimensions of  transactions $|\mathcal{X}|$ of the three datasets are listed. The weighted parameters $\omega^{\prime}$ and $\omega$ values for the datasets used in the experiments are given.}\label{dataset}
\begin{center}

\begin{tabu}{ccccc}\tabucline[1pt]{-}
  dataset& $|\mathcal{T}|$ & $|\mathcal{X}|$ & $\omega^{\prime}$ & $\omega$\\\hline
  Synthetic & $969,223$ & $4,411$ & 0.9 & 0.7 \\
  Kosarak & $990,002$ & $41,270$ & 0.9 & 0.7\\
  BMS-POS & $515,597$ & $1,658$ & 0.7 & 0.5 \\

\tabucline[1pt]{-}
\end{tabu}

\end{center}
\end{table}

\textbf{Synthetic datasets.} We generate one synthetic dataset by the IBM Synthetic Data Generation Code. Specifically, there are one million transactions were generated with 5000 categories.

\textbf{Real datasets.} We use two real-world datasets for frequent itemset mining\cite{spmf}, Kosarak and BMS-POS, which have been used in \cite{a2}.

\textbf{Metrics.} To measure the performance, we use two universal metrics in the literature, the Normalized Cumulative Rank (NCR) and the Squared Error (Var) \cite{a2}.

\begin{itemize}
\item \textbf{NCR.} It evaluates the score of itemsets identified as well as their rank. Specifically, let $P=\{X_1,X_2,...,X_k\}$ and $\tilde{P}=\{\tilde{X}_1,\tilde{X}_2,...,\tilde{X}_k\}$ denote the real and estimated top-$k$ itemsets ranked in descending order with respect to their frequencies, respectively. A quality function $q(\cdot)$ for a given itemset is defined as its score of rank, i.e., $q(\tilde{X}_i)=(k-i+1)$. All other itemsets that not in $P$ have a score of 0. Then the NCR is defined as follows:
\begin{equation}
\operatorname{NCR} = \frac{\sum_{X\in \tilde{P}}q(X) }{\sum_{X\in P}q(X)} .
\end{equation}
where the denominator is constant and equal to $\frac{k(k+1)}{2}$.
\item \textbf{Var.} It measures the estimation accuracy. For itemset $X \in P \cap \tilde{P}$, let $f(x)$ and $\tilde{f}(x)$ be the real and estimated frequencies, respectively, we have
\begin{equation}
\operatorname{Var} = \frac{1}{|P \cap \tilde{P} |} \sum_{x \in P \cap \tilde{P}} (f(x) - \tilde{f}(x))^2 .
\end{equation}
\end{itemize}

Note that in our experiments, all experimental results are repeated 20 times in order to eliminate the randomness caused by the LDP setting. We set $\xi=3$ for function $\operatorname{CutdownCandidate}(.)$ for all different datasets. Both prefix and itemset weighted parameters $\omega^{\prime}$ and $\omega$ are set in term of datasets as shown in Table~\ref{dataset}. The ratio threshold $\theta_0$ is set to $0.3$ all the time.

\subsection{Overall Results}
\label{result}
Now we compare the performance between LDP-FPMiner and SVSM. Specifically, we evaluate the NCR and Var metrics of discovering length-$\alpha$ itemsets ($\alpha \geq 2$) over three datasets when only $\epsilon$ varies and when only $k$ varies (the length-$1$ itemsets, that is, the $k$ frequent items, have been collected privately through the same protocol). Besides, for LDP-FPMiner, we present the performance of a key optimization, i.e., the itemset weighted combination.

\subsubsection{\textbf{The impact of $\epsilon$}}
The results on three datasets when only $\epsilon$ varies are presented in Fig. \ref{fig:vary ep} (the NCR metric on the first line, and the Var metric on the second line). In almost all settings, the LDP-FPMiner has the significantly higher NCR values than SVSM. This means that LDP-FPMiner is more accurate than SVSM. The advantage is more obvious when the $\epsilon$ values are small. As $\epsilon$ increases, the advantage gradually decreases. Similarly, the LDP-FPMiner has much smaller Var values than SVSM when the $\epsilon$ is small, which means that the noise added by LDP-FPMinder is effectively reduced, and the advantage decreases as the $\epsilon$ becomes large. Surprisingly, even the curves of LDP-FPMiner when $k=100$ perform better than those of SVSM when $k=50$. For different datasets, it appears both schemes work best for the Synthetic dataset. Despite of this, LDP-FPMiner has still an obvious advantage over SVSM over this dataset. In short, LDP-FPMiner introduces much less noise, and thus outperforms SVSM significantly when $\epsilon$ is small.


\begin{figure*}[tb]
  \centering
  \subfigure[NCR, Synthetic]{\includegraphics[width=0.3\textwidth]{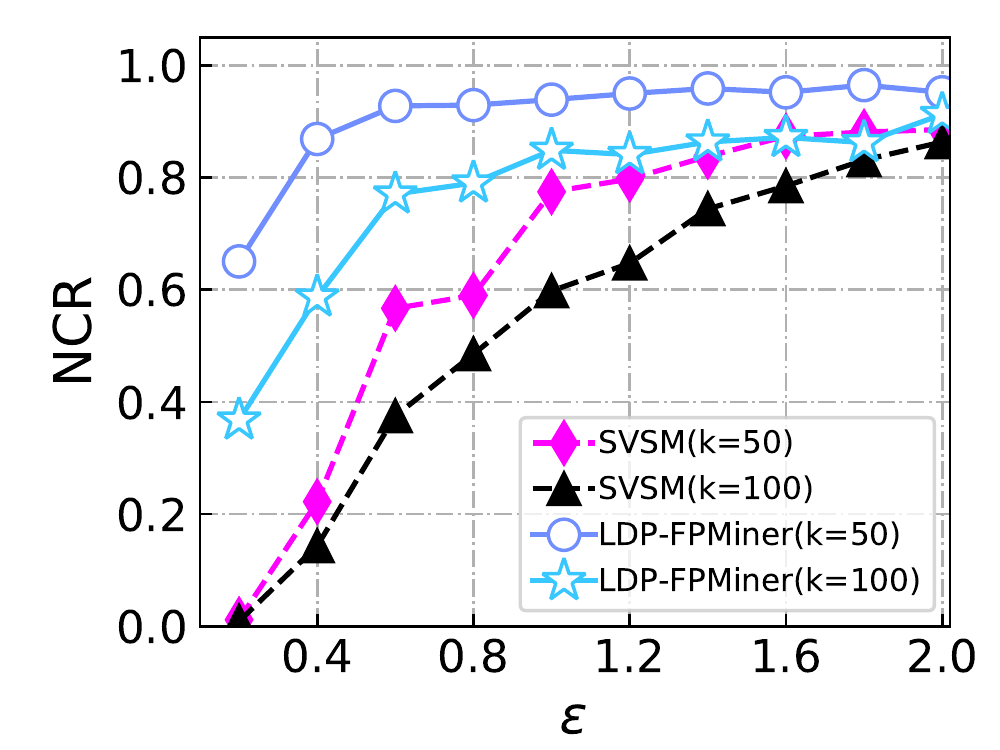}\label{ncr vary ep syn}}
  \subfigure[NCR, Kosarak]{\includegraphics[width=0.3\textwidth]{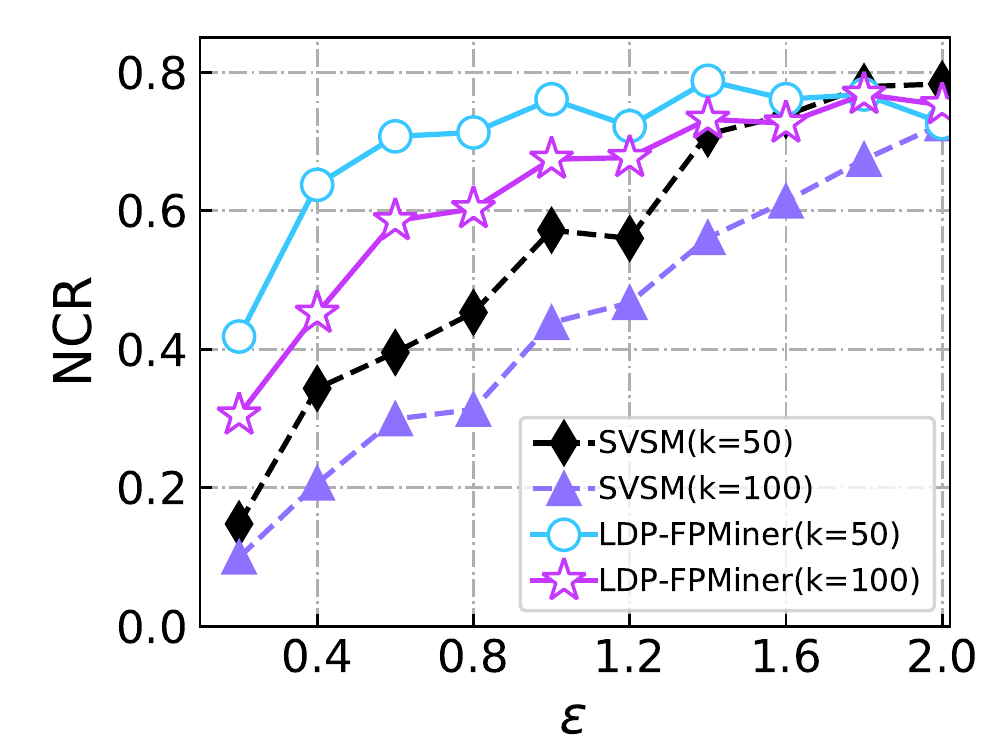}\label{ncr vary ep kos}}
  \subfigure[NCR, BMS-POS]{\includegraphics[width=0.3\textwidth]{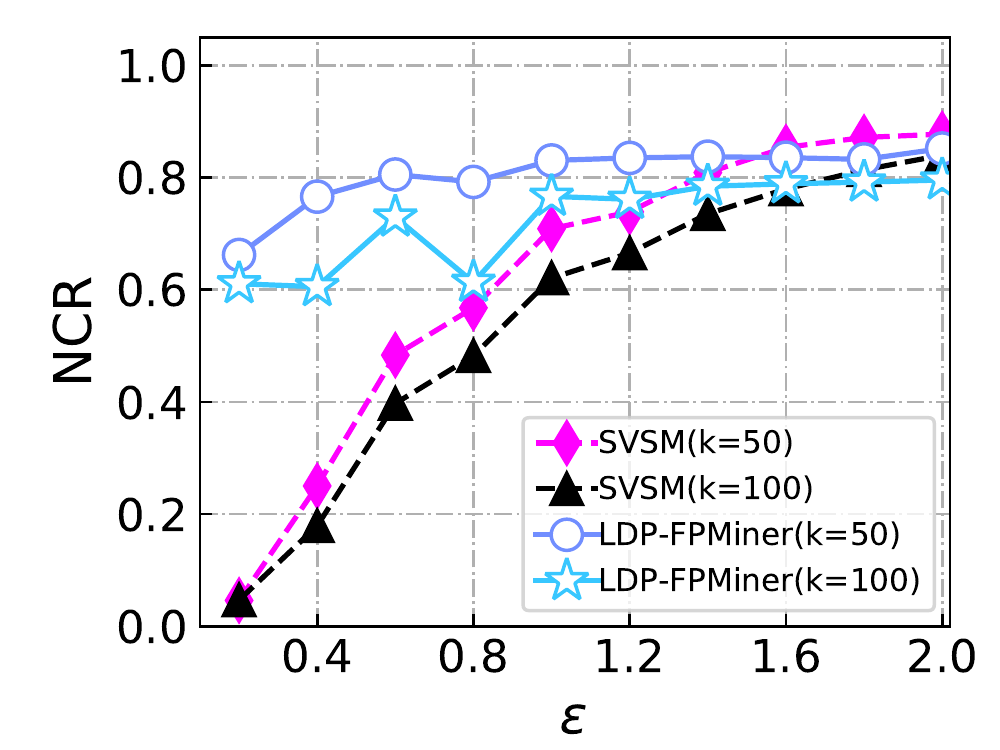}\label{ncr vary ep pos}}

  \subfigure[Var, Synthetic]{\includegraphics[width=0.3\textwidth]{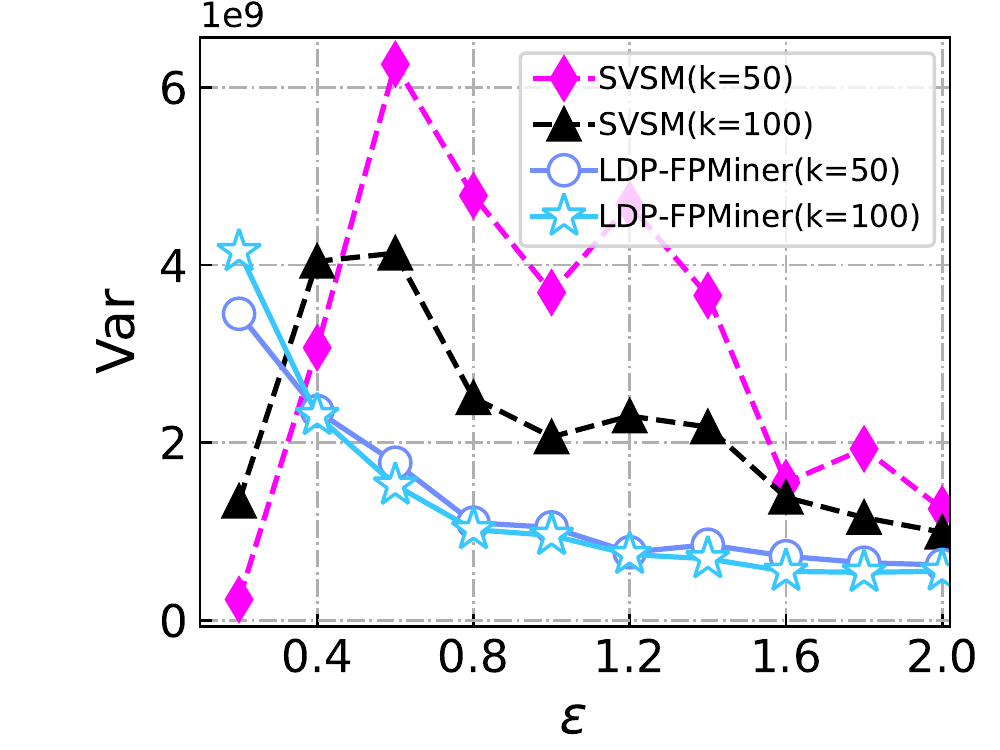}\label{re vary ep syn}}
  \subfigure[Var, Kosarak]{\includegraphics[width=0.3\textwidth]{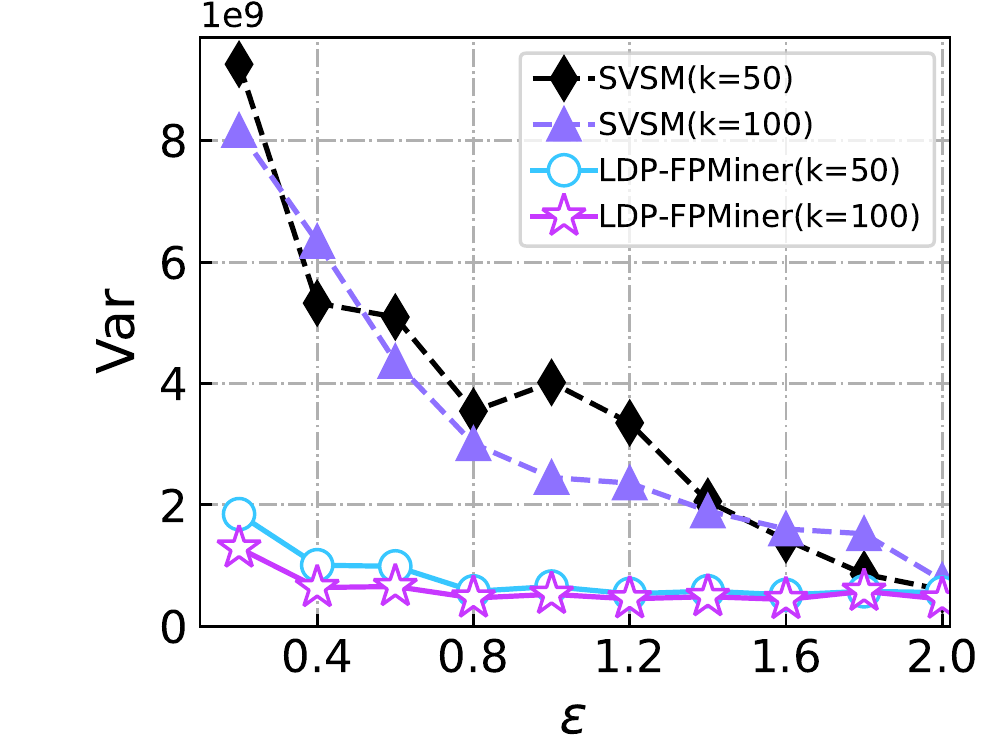}\label{re vary ep kos}}
  \subfigure[Var, BMS-POS]{\includegraphics[width=0.3\textwidth]{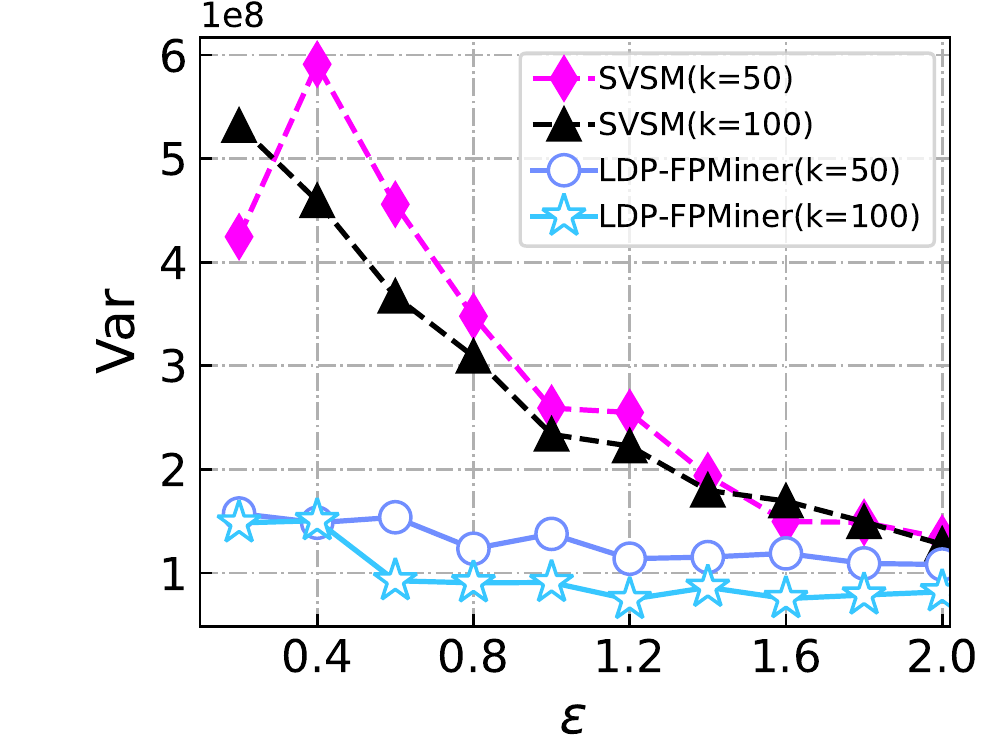}\label{re vary ep pos}}
\caption{Performances when varying $\epsilon$ and fixing $k = 50$ and $100$.}
\label{fig:vary ep}
\end{figure*}

\begin{figure*}[tb]
  \centering
  \subfigure[NCR, Synthetic]{\includegraphics[width=0.3\textwidth]{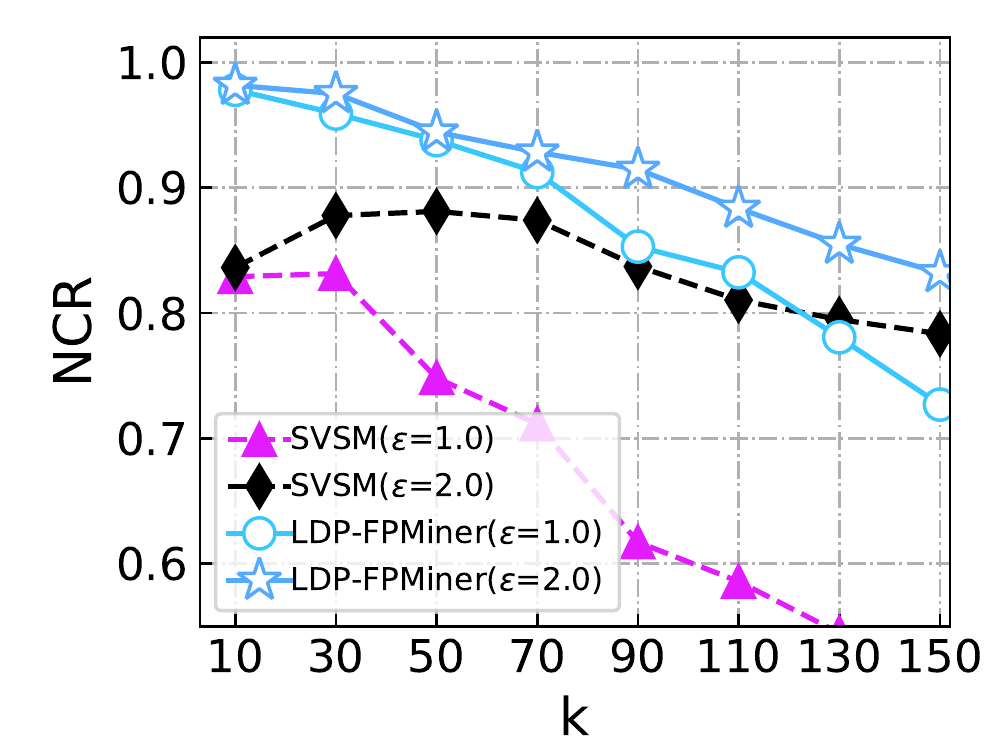}\label{ncr vary k syn}}
  \subfigure[NCR, Kosarak]{\includegraphics[width=0.3\textwidth]{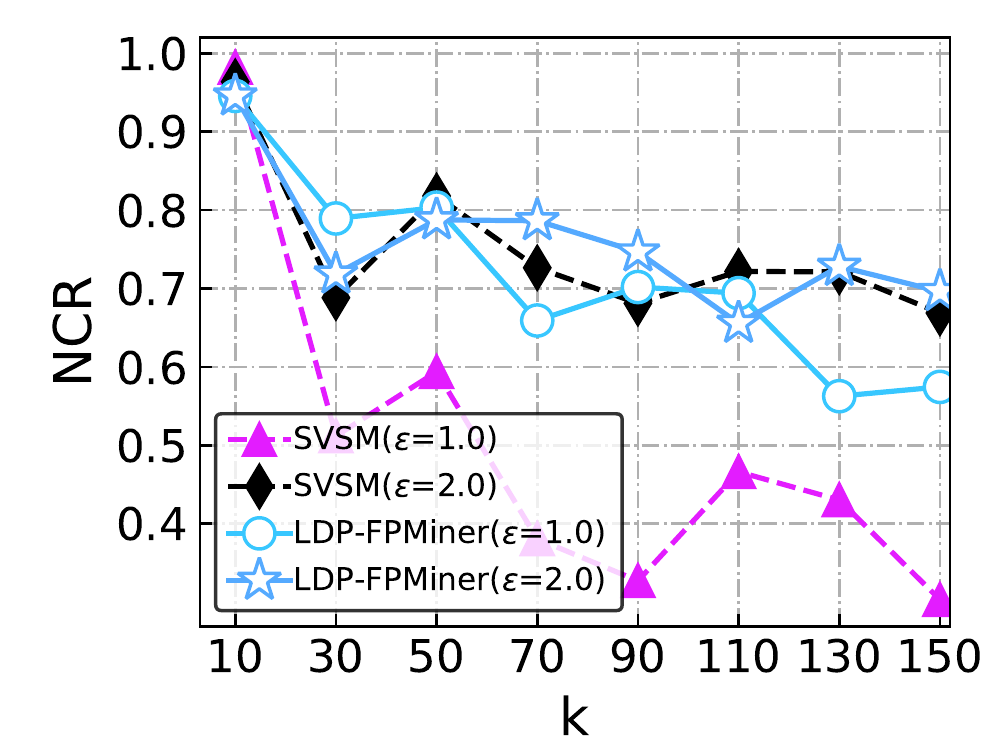}\label{ncr vary k kos}}
  \subfigure[NCR, BMS-POS]{\includegraphics[width=0.3\textwidth]{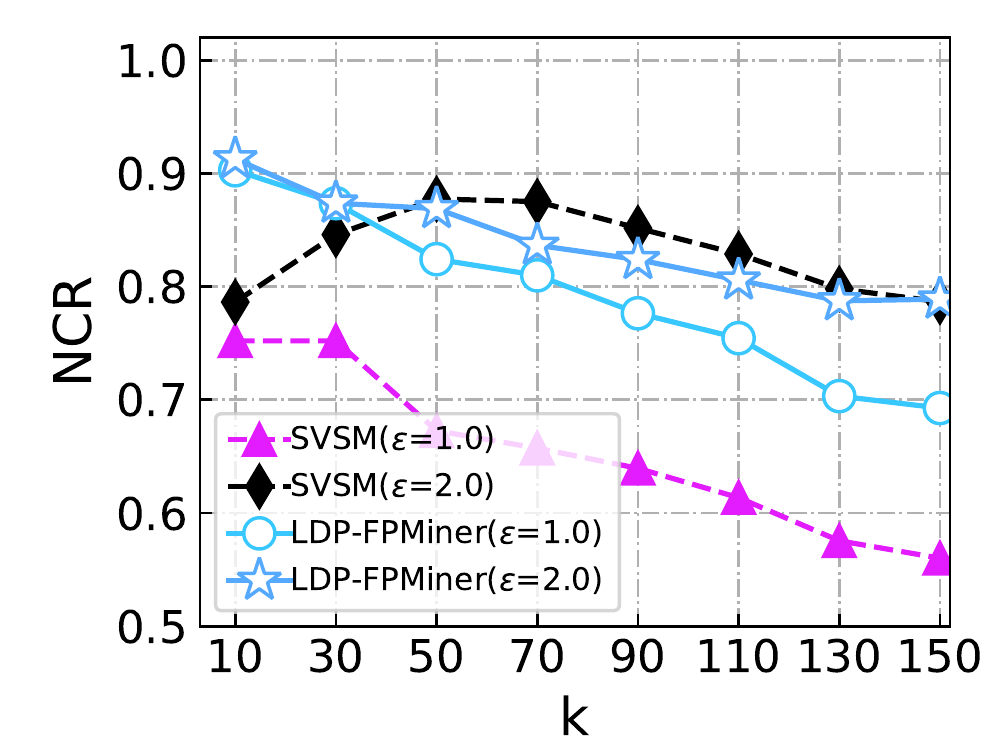} \label{ncr vary k pos}}

  \subfigure[Var, Synthetic]{\includegraphics[width=0.3\textwidth]{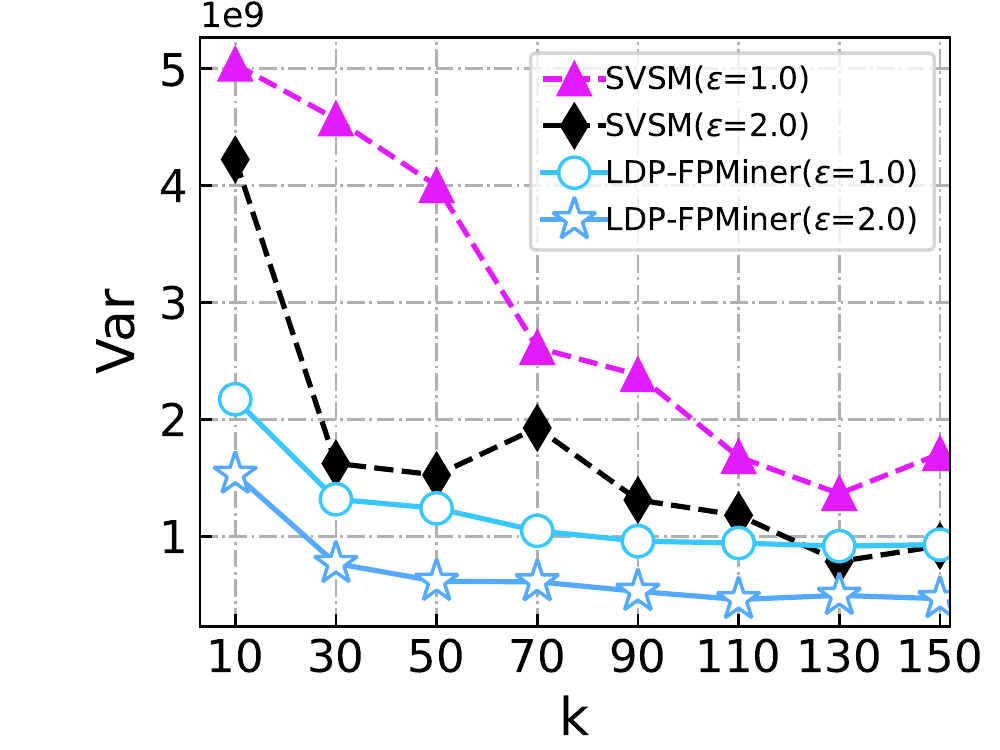}\label{re vary k syn}}
  \subfigure[Var, Kosarak]{\includegraphics[width=0.3\textwidth]{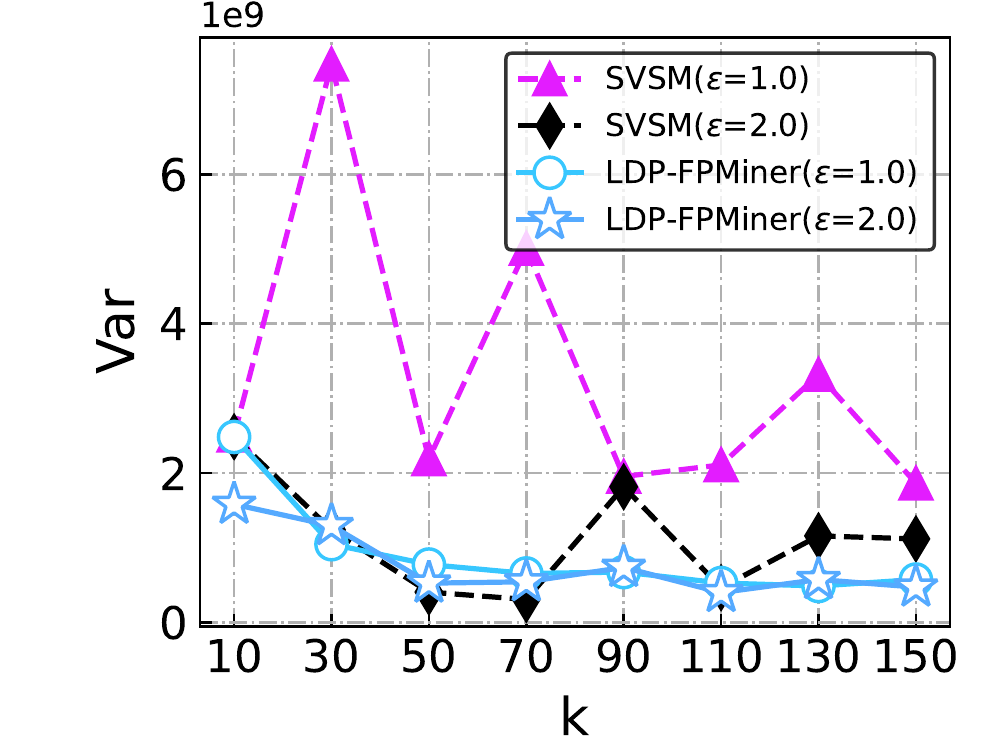}\label{re vary k kos}}
  \subfigure[Var, BMS-POS]{\includegraphics[width=0.3\textwidth]{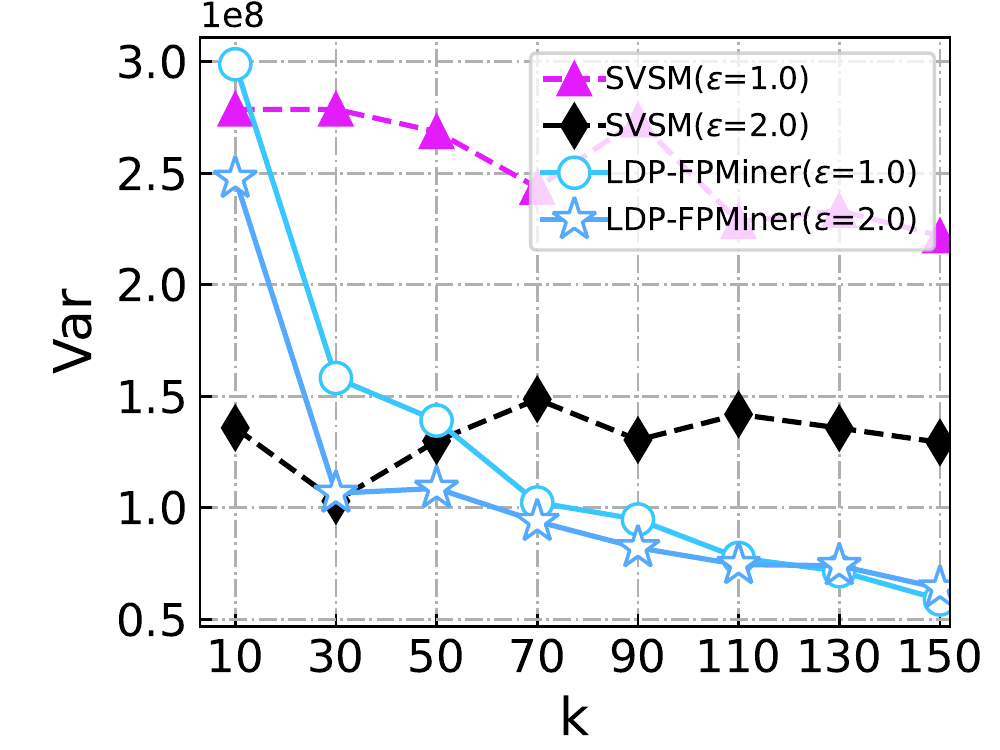} \label{re vary k pos}}
\caption{Performances when varying $k$ and fixing $\epsilon = 1$ and $2$.}
\label{fig:vary k}
\end{figure*}

\begin{figure*}[tb]
  \centering
  \subfigure[Synthetic]{\includegraphics[width=0.3\textwidth]{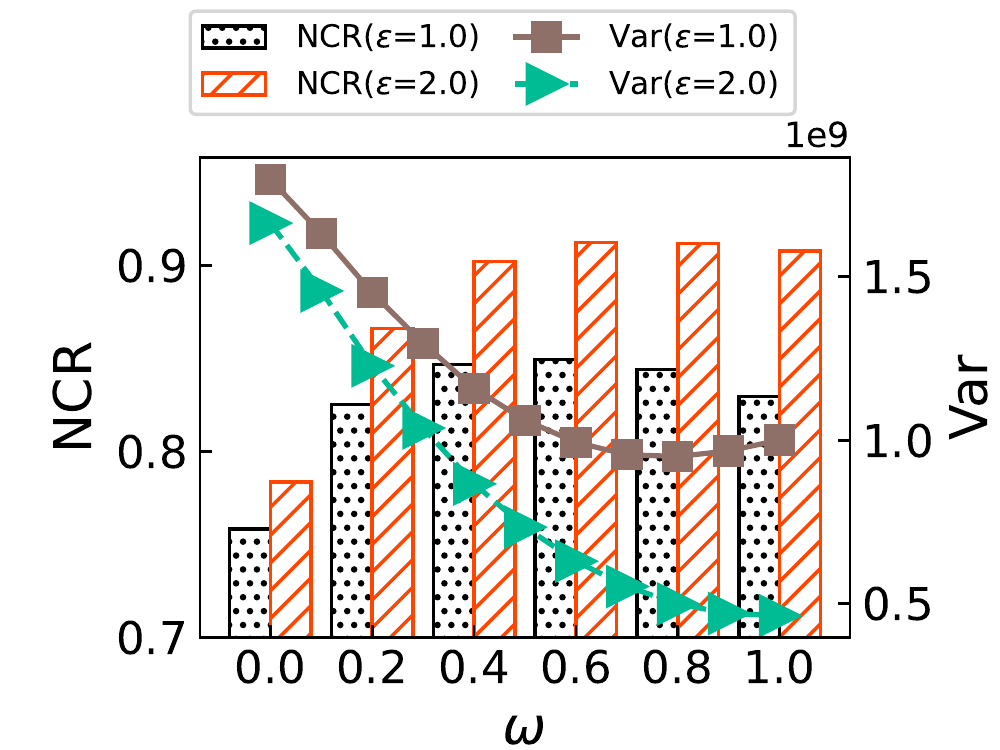}\label{ncr weight syn}}
  \subfigure[Kosarak]{\includegraphics[width=0.3\textwidth]{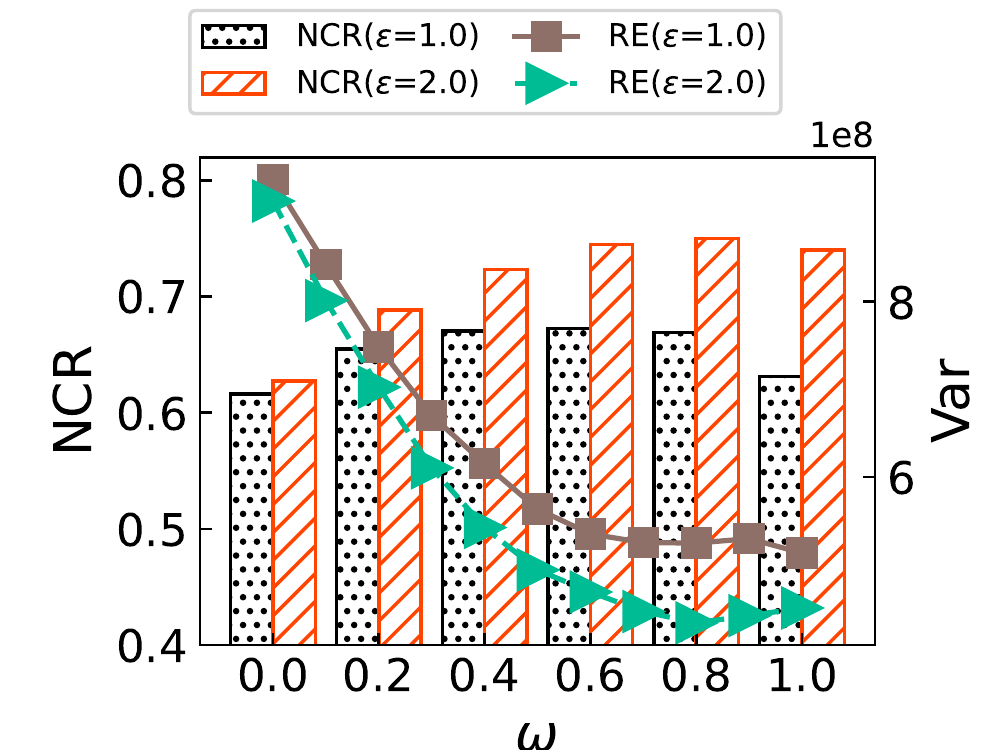}\label{ncr weight kos}}
  \subfigure[BMS-POS]{\includegraphics[width=0.3\textwidth]{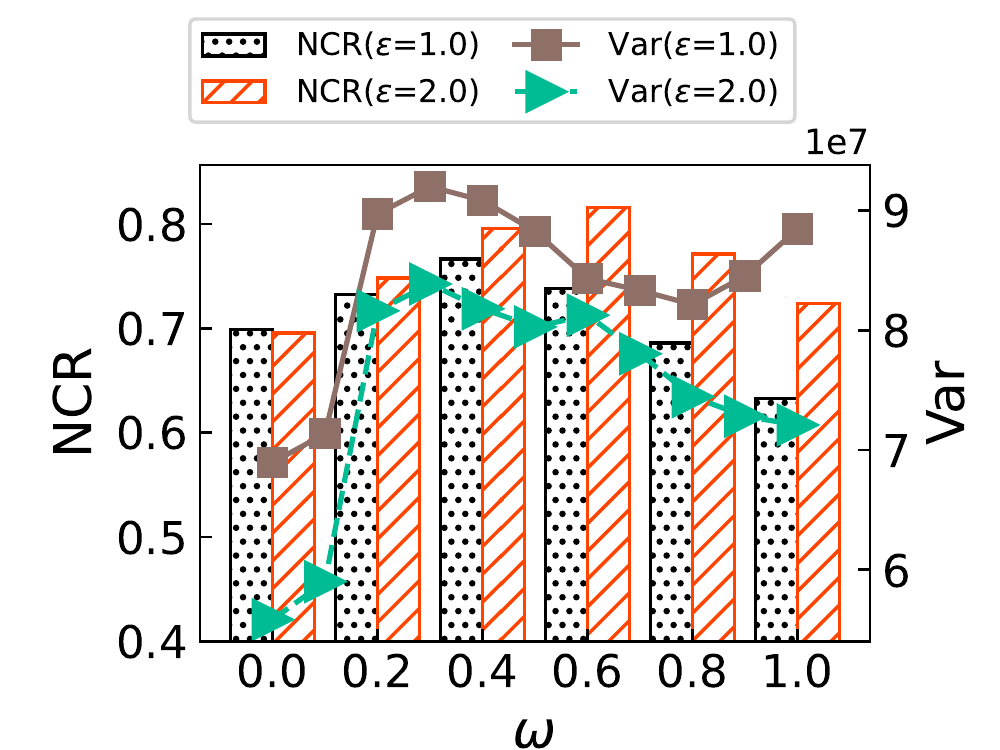}\label{ncr weight pos}}
\caption{Effectiveness of the itemset weighted combination when fixing $k = 100$ and $\epsilon=1$ and $2$.}
\label{fig:vary weight}
\end{figure*}

\subsubsection{\textbf{The impact of $k$}}
Both the results of NCR (on the first line) and Var (on the second line) are presented in Figure \ref{fig:vary k} when setting $\epsilon = 1$ and 2 separately (the privacy settings in deployed Apple protocol\cite{apple}). In all the settings where $\epsilon = 1$, LDP-FPMiner has the significantly higher NCR values than SVSM. In the setting where $\epsilon = 2$, LDP-FPMiner has higher NCR values than SVSM for both Synthetic and Kosarak datasets, but has slightly lower NCR values for BMS-POS dataset at some points. It seems that when $\epsilon$ is big enough, the SVSM scheme is still very effective, and comparable to the LDP-FPMiner. However, when $\epsilon = 1$ for all the datasets, LDP-FPMiner greatly outperforms SVSM. In the view of metric Var, LDP-FPMiner has lower Var values than SVSM in almost all the cases. It is surprising that even when LDP-FPMiner has slightly lower NCR values than SVSM, the corresponding Var values are still lower than SVSM. This again indicates that LDP-FPMiner is very effective for reducing the noise. 

In conclusion, LDP-FPMiner outperforms the SVSM in the top-$k$ task of FIM in the context of LDP. More specifically, in the case when $\epsilon$ is small (e.g. smaller than 2), it achieves higher score of itemsets identified as well lower noise injected on large domain datasets.

\subsubsection{\textbf{The impact of $\omega$}}
The effectiveness of the itemset weighted combination (explained in Section \ref{optimizations}) over three datasets when fixing $k = 100$ are presented in Figure \ref{fig:vary weight}. We use the version with all optimizations to illustrate the effectiveness. It turns out that this optimization effectively improves the performance of LDP-FPMiner, where the original result is when $\omega = 1$. As shown in Fig. \ref{fig:vary weight}, the selection of $\omega$ should balance between accuracy and error, and we give the reference selection that used in this paper as shown in Table \ref{dataset}.

\subsection{Optimizations}
\label{sec:optimizations}

In this subsection, we first compare LDP-FPMiner only applying a single optimization with the original version (without any optimizations). Then, we compare LDP-FPMiner simultaneously applying two or more optimizations with the original version. These experiments illustrate the effectiveness of single or combined optimizations.

The otimizations are introduced in Section \ref{optimize}, and for convenience we list both their abbreviations and full names as below:
\begin{itemize}
\setlength{\itemsep}{0pt}
\item PWC: prefix weighted combination,
\item CCI: conditional constrained inference,
\item NPB: negative-positive balance,
\item IWC: itemset weighted combination.
\end{itemize}

These optimizations will be applied seperately or simultaneously to LDP-FPMiner, and the BMS-POS dataset is used for evaluation.

\subsubsection{\textbf{Single Optimizations}}

We apply a single optimization to the orignal LDP-FPMiner each time, and show the effectiveness of each optimization seperately. Figure~\ref{fig:single} traces both NCR and Var values for each single optimizations. We can see that both PWC and CCI optimizations are seperately effective to improve the NCR values, namely the accuracy to identify frequent itemsets, but they cause higher Var values, namely the squre error. On the contrary, the NPB causes slightly lower NCR values than the original, but it reduces the noise very effectively. The last optimization IWC works rather well, and it raises NCR values and reduces Var values, improving the both metrics greatly. In a word, these optimizations tend to be combined together to get a good performance with high NCR and low Var values.

\begin{figure*}[tb]
  \centering
  \subfigure[NCR]{\includegraphics[width=0.32\textwidth]{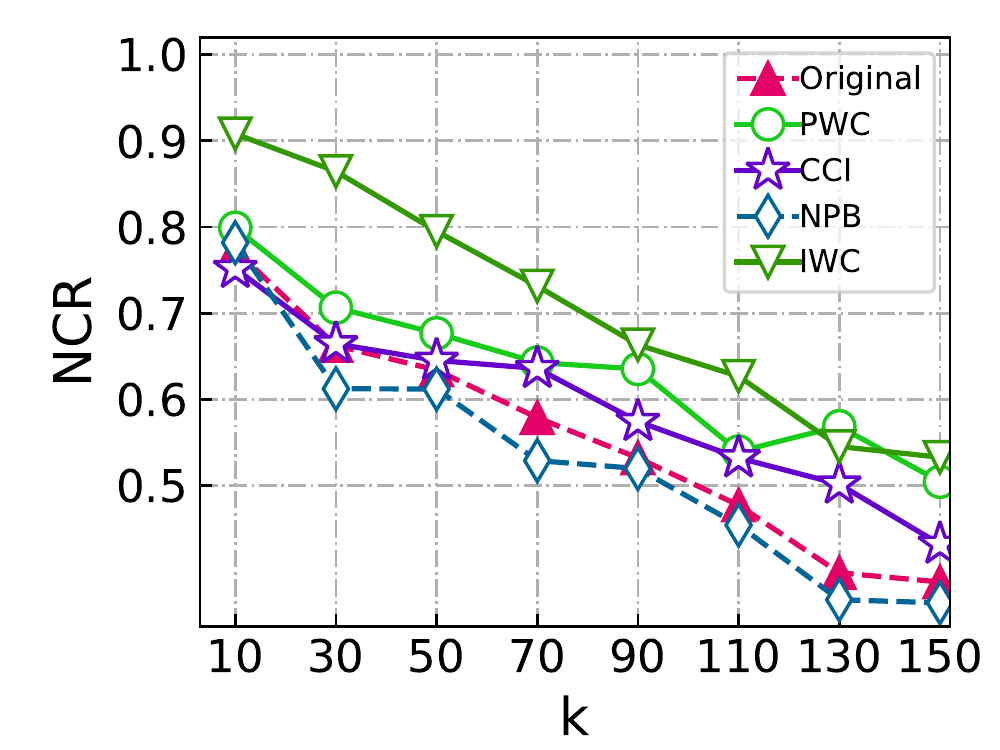}\label{fig:single_ncr}}
  \subfigure[Var]{\includegraphics[width=0.32\textwidth]{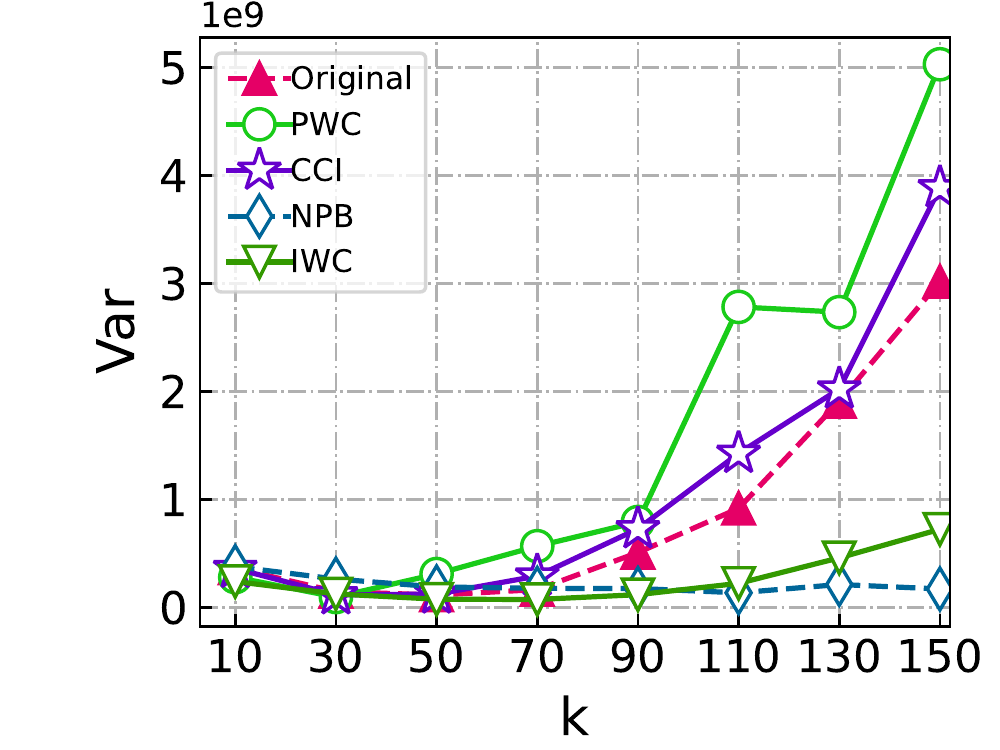}\label{fig:single_var}}
\caption{Effectiveness of single optimizatons when fixing $\epsilon=1$ and the dataset is BMS-POS.}
\label{fig:single}
\end{figure*}

\subsubsection{\textbf{Combined Optimizations}}

We gradually combine two or more optimizations together and apply them to LDP-FPMiner. Figure~\ref{fig:combined} illustrates the effectiveness of combined optimizations. We can see that the combination of PWC and CCI optimizations increases the NCR values, but also increases the Var values. However, when we combine PWC, CCI and NPB optimizations, it is surprising that the NCR values are further increased and the Var values are greatly reduced and become lower than those of the original version. Finally, we further combine all optimizations PWC, CCI, NPB and IWC together, and obtain the final result with even higher NCR values and lower Var values. It appears that the combination of optimizations magnifies the improvements. The underlying reason may be that all optimizations are on the right way to identify frequent itemsets and reduce frequency noise added, and they are complementary to each other. Additionally, for optimizations PWC and CCI, the improvement over NCR values gradually become more significant as $k$ increases, while for optimizations NPB and IWC, the Var values are reduced more greatly as $k$ increases. It may be because there are more room to improve when $k$ is large.

\begin{figure*}[tb]
  \centering
  \subfigure[NCR]{\includegraphics[width=0.32\textwidth]{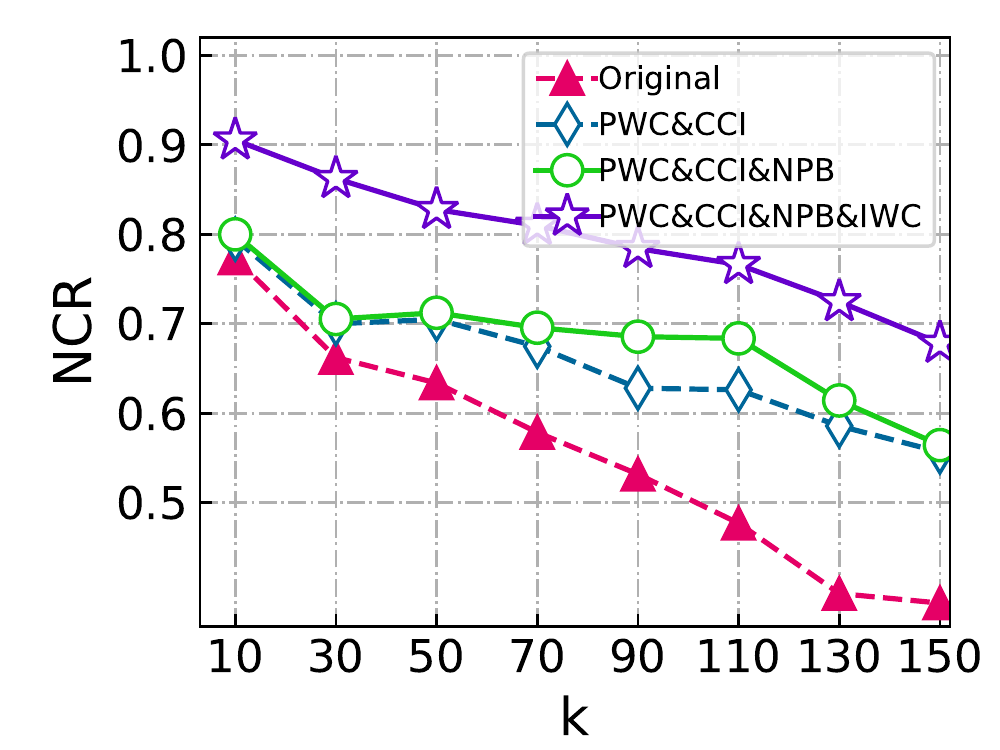}\label{fig:combined_ncr}}
  \subfigure[Var]{\includegraphics[width=0.32\textwidth]{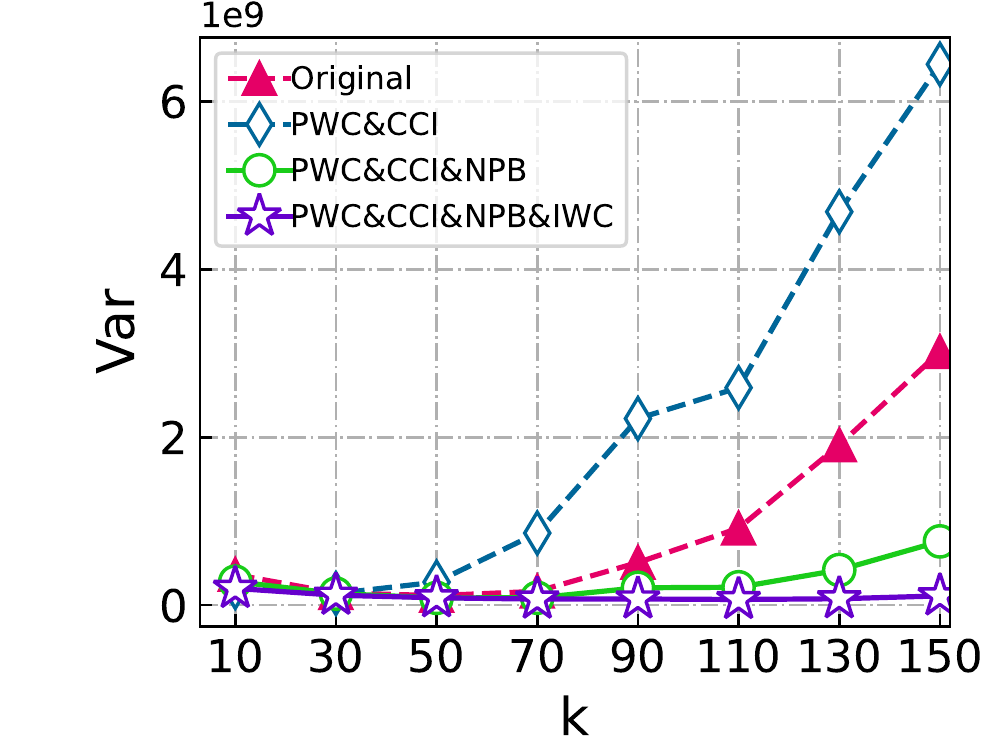}\label{fig:combined_var}}
\caption{Effectiveness of combined optimizatons when fixing $\epsilon=1$ and the dataset is BMS-POS.}
\label{fig:combined}
\end{figure*}

\section{Related Work}\label{relatedwork}
Local differential privacy (LDP) has become more and more popular for data privacy preservation, and most of existing works focus on basic statistics (e.g. \cite{a8,rr,rappor,a12,b3,imola2022communication}) to estimate mean values over numeric attributes or frequencies over categories. Besides, in recent years, there are many more complicated statistical analyses (e.g. heavy hitters\cite{privtrie,b3,bun2019heavy,cormode2021frequency}, key-value collection\cite{privkv,gu2019pckv}, multidimensional data\cite{b1,alibaba,xu2020collecting} and set-valued data\cite{privset,a1,a2} analysis) are proposed using frequency estimation as a building block.

Due to the set property of the data, the transactional (or set-valued) data setting is more challenging even when one just tries to find heavy hitters, not mention discovering itemsets. In the particular LDP setting, Qin et al.\cite{a1} propose the LDPminer that discover heavy hitters in two phases and leave FIM problem as an open problem. In \cite{privset}, the set-valued data aggregation mechanism PrivSet is proposed with low computational overhead but does not work well when the domain is large. To the best of our knowledge, the state-of-the-art solution \cite{a2} to FIM identified itemsets based on the PSFO protocol, and did not consider frequency consistency among itemsets. In this paper, we propose and optimize the FP-tree based approach by exploiting frequency consistency, and identify frequent itemsets effectively with high accuracy and low noise.

Besides, in the centralized differential privacy (CDP) setting, Bhaskar et al.\cite{a3} propose an approach with the exponential mechanism as well as the Laplace mechanism to release top-$k$ itemsets of length not greater than predefined factor $m$. Li et al.\cite{a4} define the $\theta$-basis set to improve the utility. The concurrent approach\cite{a5} improves the trade-off between privacy and utility with smart truncating as well as double standards. Lee et al.\cite{a6} identify top-$k$ itemsets and then construct a compact, differentially private FP-tree to derive frequencies of itemsets. These works are quite different from ours for the raw data from users are available in the CDP setting.

\section{Conclusion}\label{conclusion}
In this paper, we study the problem of privacy-preserving frequent itemset mining, and discover $k$ most frequent itemsets from sensitive transactions with LDP. The state-of-the-art protocol SVSM mainly applies the idea of guessing frequencies to find the candidate itemsets and then further identifies the top-$k$ itemsets with frequency oracle protocol without considering frequency consistency. Different from this, we combine frequent pattern tree (FP-tree) method, frequency oracle protocol, and guessing frequencies to build and optimize a noisy FP-tree with LDP by exploiting frequency consistency among itemsets, and then mining this FP-tree to find the top-$k$ frequent itemsets. To the best of our knowledge, this is the first time that FP-tree is applied in LDP setting to mine frequent itemsets. The experimental results show that the proposed approach LDP-FPMiner outperforms the SVSM significantly.


\bibliographystyle{ACM-Reference-Format}
\bibliography{reference}
\end{document}